\documentclass[reprint,superscriptaddress,showkeys,amsmath,amsthm,amssymb,aps,floatfix,prx]{revtex4-2}
\usepackage{amsmath}
\usepackage{amsthm}
\usepackage{amsfonts,amssymb}
\usepackage{ifthen}
\usepackage{relsize}
\usepackage{graphicx}
\usepackage{xcolor}
\usepackage{float}
\usepackage{hyperref}
\usepackage{caption}
\usepackage{subcaption}
\usepackage{comment}
\usepackage{todonotes}

\newtheorem{definition}{Definition}
\newtheorem{remarque}{Remark}%
\newtheorem{theorem}{Theorem}
\newtheorem{corollaire}{Corollary}%
\newtheorem{lemma}{Lemma}
\newtheorem{proposition}{Proposition}
\newcommand{\R}{\mathbb{R}}
\newcommand{\N}{\mathbb{N}}

\newcommand{\X}{\mathcal{X}}
\newcommand{\bint}{\{l,r\}^*}
\definecolor{darkgreen}{RGB}{0,128,43}

\newcommand{\com}[1]{}
\newcommand{\comb}[2]{\binom{#2}{#1}}
\newcommand{\corrA}[1]{{#1}}
\newcommand{\corrB}[1]{{#1}}
\newcommand{\identity}{\text{Id}}
\newcommand{\VL}[1]{#1}

\DeclareMathOperator{\indice}{ind}

\usepackage{xcolor}

\input{macro_nodes_inline}
\newcommand{\gcirc}{X}
\newcommand{\interb}{I_2}
\newcommand{\setcycle}[1]{\mathcal{C}_#1}
\newcommand{\hs}{\sqrt{\tau}}
\newcommand{\inter}{I}
\newcommand{\shift}{\tau}

\begin{document}

\title{Time arrow without past hypothesis\corrA{: a toy model explanation}} %

\author{Pablo Arrighi}
\author{Gilles Dowek}
\affiliation{%
Université Paris-Saclay, Inria, CNRS, LMF, 91190 Gif-sur-Yvette, France
}%

\author{Amélia Durbec}
\affiliation{%
CNRS, Centrale Lille, JUNIA, Univ. Lille, Univ. Valenciennes, UMR 8520 IEMN, 59046 Lille Cedex, France
}%

\begin{abstract}
The laws of Physics are time-reversible, making no qualitative distinction between the past and the future---yet we can only go towards the future. This apparent contradiction is known as the `arrow of time problem'. Its current resolution states that the future is the direction of increasing entropy. But entropy can only increase towards the future if it was low in the past, and past low entropy is a very strong assumption to make, because low entropy states are rather improbable, non-generic. Recent works from the Physics literature suggest, however, we may do away with this so-called `past hypothesis', in the presence of reversible dynamical laws featuring expansion. We prove that this \corrA{can be} the case in principle, \corrA{within} a toy model. It consists in graphs upon which particles circulate and interact according to local reversible rules. Some rules locally shrink or expand the graph. We prove that almost all states expand; entropy always increases as a consequence of expansion---thereby providing a local explanation for \corrA{the rise of an entropic} arrow of time without the need for a past hypothesis. The discrete setting of this toy model allows us to deploy the full rigour of theoretical Computer Science proof techniques. It also allows for the numerical exploration \corrA{of several physically-motivated variants: a time-symmetric variant; two inflationary variants; and a damping variant---which slows down thermal death. The fact that all of these models exhibit similar behaviours suggests that local reversible expansion mechanisms constitute a robust recipe for a time arrow without past hypothesis. In this qualitative sense, the explanation may therefore also be relevant at the  cosmological level.}  
\end{abstract}

\maketitle

\section{Introduction}
\paragraph{In short.} The main contribution of this paper is a first rigorous proof that an arrow of time can emerge from generic configurations under some time-reversible laws, without the need for a past hypothesis. This problem pertains to a long `foundations of Physics tradition'. The novelty of our approach lies in the deployment of theoretical Computer Science models and techniques, in order to provide an in-principle proof that this is possible. I.e. the problem is first transposed to a well-controlled toy model (in terms of graph dynamics), where the statement is well-formalised (in terms of entropy growth). It is then addressed through the study of invariants and termination proofs.
Next we bringing in more physical considerations, that naturally drive us to explore more complex variants of the model, e.g. in order to cater for time-symmetry; account for an inflationary period; or be able to witness the arrow of time locally. Numerically, we find that these variants also feature a time arrow without past hypothesis, and that `inelastic shocks' help fight thermal death. The latter phenomenon also admits an intuitive explanation relying on a conservation law.

\paragraph{The problem.} Physics laws are time-reversible, i.e. time evolution can be inverted,  making no qualitative distinction between the past and the future. \corrA{At least this is the case of classical mechanics, general relativity, and for quantum mechanics devoid of measurements, e.g. if all systems are considered to be quantum mechanical.}\\
Yet, we clearly experience the fact that we cannot go back to the past. This discrepancy is referred to as the `arrow of time problem'. It sets the requirement for an explanation of ``how does the future versus past phenomenon that we witness in everyday life, arise from time-reversible dynamical laws alone''. I.e. the aim of the game is to resolve this apparent paradox by pinpointing specific physical quantities that can be used as `clocks', and then identifying the future direction to be the direction of increase of that quantity, thereby providing a time arrow.

Usually, these quantities are variants of the concept of entropy, and one shows that starting from a low entropy initial state, entropy typically increases even under a time-reversible law. Thus, in the conventional argument due to Boltzmann \cite{Boltzmann}, ``Time-reversible dynamical law + Past low entropy $\Rightarrow$ Arrow of time''. This argument works even for systems of bounded size (e.g. a few of particles in a box). But it suffers three important drawbacks: $(i)$ By the Poincaré recurrence theorem, as we iterate the dynamics forward, the entropy typically increases\ldots but then drops back, and increases, etc., as the dynamical system is necessarily (almost-)periodic. This problem is usually dismissed on the more practical ground that the recurrence period is beyond cosmological times, but in absolute terms this remains a valid criticism.
 $(ii)$ Starting from an initial low entropy state and iterating the dynamics forward, entropy typically increases (the `entropic clock's arrow' matches that of the external time-coordinate aka ``dynamical clock''). However, had we applied the reversed dynamics instead, i.e. iterating the dynamics backward, entropy typically would have increased, too (the entropic clock's arrow does not match the dynamical clock's arrow) \cite{GoldsteinJanusGlass}, as is quite often overlooked. $(iii)$ The assumption that the initial state be of low entropy, also referred to as `the past hypothesis' \cite{AlbertPastHypothesis} is a very strong one. This is because generic states have maximal entropy. Hence, the presence of such an improbable state at dynamical clock time $0$, remains a mystery, that again demands an explanation. We will review these points in Sec.~\ref{sec:conventional}. 

\paragraph{Comparison with related models.} In \cite{CarrollChen,CarrollChen2}, Carroll and Chen try to fix the third criticism and provide the first plausible intuition why ``Time-reversible dynamical law $\Rightarrow$ Arrow of time''. Their key new ingredient is big bounce (big crunch then bang) then eternal expansion---as featured also in several early attempts to restore the particle-level charge/parity/time-reversal symmetry (CPT) at the scale of the universe \cite{Sakharov,BoyleCPT,HartleHawking}. Their discussion is left informal however, leaving plenty of room to discuss whether it really manages to do away with the past hypothesis \cite{WaldVersusChenCarroll,VilenkinAgainstCarrollChen,WeaverAgainstCarrollChen}. In particular, they themselves raise the issue whether expansion mechanisms are actually compatible with reversibility.

Informally, in a big bouncing universe, matter gets compressed by the big crunch, and released by the big bang. Expansion happens so fast that matter then finds itself out of equilibrium, in a low entropy state. Matter then diffuses and entropy increases without ever reaching a maximum, as expansion is eternal. The entropic clock's arrow thus matches that of the ``size-of-the-universe clock'', which in turn matches that of the dynamical clock after the big bounce. This explanation is compelling. Yet the following direct logical consequence is somewhat mind-bending: before the big bounce occur, the size-of-the-universe clock, and thus the entropic clock, have their arrow opposite to that of the dynamical clock. %
These twin entropic times facing each other have in fact been popularized as `Janus time' by Barbour \cite{BarbourNBody}. Still, as counter-intuitive as it may be, this is already happening in the conventional argument by Boltzmann: as pointed out by Golstein et al. \cite{GoldsteinJanusGlass}, this is just the mentioned criticism $(ii)$. %

Barbour et al. \cite{BarbourNBody,BarbourAgainstZeh} use the $n-$body problem, with (non-local) Newtonian classical gravity turned on, as an enlightening analogy of big-bounce-then-eternal-expansion. The question whether their model really does away with the past hypothesis has been argued in Zeh \cite{ZehAgainstBarbour}, \corrB{for whom: ``the numerically studied examples are not at all typical''}. Indeed as the considered bodies travel on a pre-existing infinite space, the analogy blurs out the requirement of finite but unbounded configurations, which is needed to make the argument rigorous. Moreover, the entropic clock (quantity measuring the microscopic disorder) is replaced by a ``shape complexity clock'' (quantity measuring the macroscopic clumping) in this strand of works. This is non-standard and arguable on the basis that in many situations there is no need for gravitational clumping in order to observe an arrow of time.
The question whether expansion can be implemented as a local reversible mechanism also remains open in this strand of works \cite{KoslowskiThroughBigBang}.

The aim of this paper is to exhibit rigorously-defined local reversible dynamical laws (the local rule of reversible causal graph dynamics) for which we can prove that:
\begin{itemize}
\item Almost all states end up growing in size as we iterate the dynamics.
\item Entropy always increases as size grows.
\end{itemize}
\corrA{This entropy is rigorously-defined in a way that is analogous to that of perfect gas. Namely, it is given by the log of the number of graphs corresponding to two macroscopic properties, namely the size of the graph and its number of particles.} Thus we prove that an entropic clock direction emerges without the need to assume past low entropy. In other words the arrow of time is established from local reversible expansion mechanisms alone, doing away with the past hypothesis. This works because size as a function of dynamical clock time is almost always $U$-shaped. As almost all states are somewhere on this $U$-curve, their size will end up growing, and their entropy will end up increasing. They will do so forever, as configurations are of finite but unbounded size. Of course almost all states have, somewhere along the dynamical clock timeline to which they belong, some states of smallest size and lowest entropy, which may be dubbed as ``initial''. These particular states are non-generic, just like the minimum of any $U$-curve is non-generic. Because the $U$-curve is due to the dynamics alone, their existence is the result of dynamics alone.

\begin{figure}%
\centering
\includegraphics[scale=1.40]{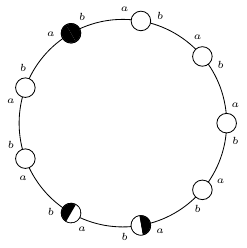}
\caption{\emph{A configuration.} \corrA{All configurations are cyclic graphs. The nodes are connected via ports $a$ or $b$. At the nodes, each} full half-disk represents the presence of a particle that is about to hop along the port \corrA{it is facing. The names identifying each node are omitted in this picture.}} 
\label{circular_graph}
\end{figure}

\begin{figure}[t]
\centering
\begin{subfigure}{0.95\columnwidth}
\centering
\includegraphics[width=0.95\columnwidth]{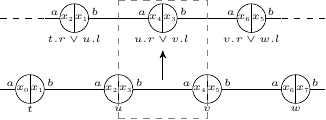}
\caption{\emph{Step $\hs $.} \corrA{Here the $x_i$ are arbitrary bits coding for the presence of a particle, or not, e.g. $x_1=1$ would mean a full half-disk at position $t$ facing the $b$ port. Right beneath the nodes are their names.} Each particle hops by half an edge. In other words, edges become vertices and vertices become edges, whilst particles follow the movement and keep their orientation. \corrA{Notice that $(\hs )^2=\tau$, i.e. two iterations produce just the partial shift where for instance $x_1$ will find itself at the $u$ node, still facing the $b$ port. Information is therefore conserved through this step.}}
\label{sqrtS_intuitive}
\end{subfigure}\\
\begin{subfigure}{0.95\columnwidth}
\includegraphics[scale=2]{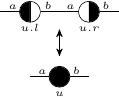}
\caption{\emph{Step $\inter$}. All occurrence of the two patterns are flipped for each other synchronously. This is unambiguous, as the patterns do not overlap. \corrA{This step is therefore involutive.}}
\label{fig:intuitive_I}
\end{subfigure}
\hfill
\caption{\label{fig:rulestoymodel}\emph{Rules of the toy model.}}
\end{figure}

Our toy model is set in $1+1$ spacetime. It consists in circular graphs upon which particles (Fig. \ref{circular_graph}) move and interact when they are closeby (Fig. \ref{fig:rulestoymodel}). During those interactions, some patterns are interchanged locally, triggering shrinking or expansion of the circle. It is cast in the framework of reversible causal graph dynamics \cite{ArrighiRCGD,ArrighiCreation} and is inspired by the Hasslacher-Meyer model \cite{MeyerLGA}, for which there is numerical evidence of a $U$-shaped size curve, but no proof---this seems inherently hard to prove in fact \cite{RabinMeyer}. There is no mention of the arrow of time nor entropy in their paper; moreover the sense in which it is reversible and causal is left informal and seems incompatible with quantum mechanics \cite{amelia_quantum_graphs}. Instead, our toy model enjoys rigorous proofs of $U$-shaped size and entropy curves, as well as rigorous notions of reversibility \cite{DurbecIMRCGD} and causality (i.e. making sure that information propagates at a bounded speed with respect to graph distance), readily allowing for a quantum extension \cite{ArrighiCreation}. These results are provided in Sec.~\ref{sec:main}.

As far as we know, there are no other closely related models besides the above-mentioned ones. This may be because 1/ The fact that a time-reversible model always grows is somewhat counter-intuitive, 2/ Let alone proving it---as shown by the efforts of \cite{RabinMeyer}. 3/ Interdisciplinary work is not so common on this topic. In particular, the theoretical Computer Science analysis that we deploy is a new player. We hope it may bring simplicity, clarity and rigour to a long-standing issue---as did many toy models in Physics. Moreover, we believe that computational nature of such a model is advantageous per se, as it allows for the numerical exploration of more complex variants, featuring more Physics. 

Indeed, much Physics is not only time-reversible but also time-symmetric, i.e. invariant under the change of variable $t=-t'$, thereby making no distinction at all between the past and future time directions. We provide numerical evidence that there is a time-symmetrized extension or our model that has $U$-shaped size curve and entropy curves in Sec.~\ref{sec:time-symmetry}.

The toy model features quadratic expansion, whereas the very early universe is believed to have experienced a short exponential expansion phase referred to as `inflation'. We provide numerical evidence that a slight variant of the toy model features an exponential $U$-shaped size curve in Sec.~\ref{sec:exponential}.

The toy model features quick `thermal death'. Once thermal thermal death is reached, the arrow of time can no longer be witnessed locally. We provide numerical evidence that just by adding inelastic collisions of matter to the model (and triggering the emission of radiation into free space so as to keep things time-reversible) thermal death gets delayed, see Sec.~\ref{sec:localentropy}.

Summary and perspectives are given in Sec.~\ref{sec:conclusion}. Notice that this article is the long, Physics-oriented journal version of a short, Computer Science-oriented conference proceeding \cite{ArrighiTimeArrowRC} focused on Sec. \ref{sec:main}.

\section{The  conventional argument}\label{sec:conventional}

{\em Entropy.} Entropy was defined by Boltzmann in the 1870s:

\begin{definition}\cite{boltzmannVorlesungenUberGastheorie1896}
The entropy $S$ of a macroscopic state is defined by :
\begin{align} S=k.\ln{\Omega}\label{eq:boltzmann}
\end{align}
with $k$ the Boltzmann constant and $\Omega$ is the number of \corrB{accessible microscopic configurations} corresponding to the macroscopic state.
\end{definition}

In this definition the word `macroscopic state' refers to a set of values for the macroscopic properties of the system, such as its temperature, pressure, volume or density. 
Given a certain macroscopic state, a `statistical ensemble' is a way to assign a probability distribution to the set of microscopic states that correspond to the macroscopic state. It is often reasonable to assume that the probability distribution be uniform (aka a `microcanonical ensemble'). The probability $p_i$ of a microscopic state $x_i$ is then $p_i=\frac{1}{\Omega}$. The connection between the Shannon entropy \cite{shannonMathematicalTheoryCommunication1948} of this probability distribution, and the Boltzmann entropy of the macroscopic state, is then obvious: \begin{align} S = - k \sum_{i=0}^{\Omega-1} p_i \ln{(p_i)} = - k \sum_{i=0}^{\Omega-1} \frac{1}{\Omega}\ln{(\frac{1}{\Omega})} = k.\ln{\Omega}\end{align}

For a dynamical system over the state space $X$, the microscopic states are simply the configurations $x_i \in X$ of the system. We formalize the macroscopic states as equivalence classes on $X$. The entropy function associates, to each microscopic configuration, the entropy of its macroscopic state.                 

\begin{definition}
Consider $X$ a set and $\equiv$ an equivalence relation on $X$. We define the entropy function $S: X \to\R$ associated with $\equiv$ as:
\begin{align}  S(x)=\ln{(|[x]|)}\end{align}
with $[x]$ the $\equiv$-equivalence class of $x$ and $|[x]|$ is its size.
\end{definition}

{\em The case of bounded size dynamical systems.}\label{sec:periodic_entropy}
The second principle of thermodynamics states that ``entropy increases in time''. 
However, even for isolated, bounded size dynamical systems, the situation is not so obvious:

\begin{remarque}\label{periodic_entropy}
Let $X$ be a finite state space and $f: X \rightarrow X$ a bijection. For any entropy function $S$, and for any configuration $x\in X$, the sequence $(S(f^{n}(x)))_{n \in N}$ is periodic because the sequence $f^{n}(x)_{n \in N}$ is periodic. This implies that the sequence of entropy variations $(S(f^{n+1}(x))-S(f^{n}(x)))_{n\in N}$ is itself periodic. If $(S(f^{n}(x)))_{n \in N}$ is not constant, then these entropy variations can be negative.
\end{remarque}

How can we justify, then, that when we dilute a drop of dye in a sealed glass of water, entropy seems to just rise, unambiguously indicating an arrow of time? Besides the fact that in practice the glass is not quite a isolated system, and hence undergoes a not quite time-reversible dynamics, several other assumptions are implicit in this emblematic experiment.

First, the duration of such an experiment will be far too short to observe periodicity. \corrB{To back this point, we remark that for a glass of water, based on an estimate of the molar entropy of liquid water at standard conditions, the Boltzmann Eq. \eqref{eq:boltzmann} yields $\Omega\sim e^{10^{25}}$. Thus, any reasonable discretization of a such macroscopic system into distinct microscopic states will lead to $\sim e^{10^{25}}$ of them. Consequently, any reasonable estimate of how long it will take to forcibly get back to one such microscopic state will again be of that order. Yet, it is estimated that only $10^{60}$ Planck times have elapsed since the Big Bang. So, the age of the universe is negligible before the estimated recurrence time.}

Second, the drop of dye would likely have diluted just as well if it had undergone the time-symmetrized versions of Physics laws instead. In other words, entropy typically does increase when we start from a low entropy initial configuration\ldots  but it does so in both directions of the dynamical clock \cite{GoldsteinJanusGlass}. 

Third, we must realise that the experiment starts at a rather improbable time: the first $10^{60}$ ticks of dynamical clock time represent a negligible fraction of the $\sim e^{10^{25}}$ entropy recurrence period. Had current time been picked up at random within the period, there would be no reason to expect it to be a time of increase of entropy, rather than of decrease. Another way to say this is that the experiment starts from an improbable configuration. Indeed in any generic configuration the dye is diluted already; entropy is almost maximal already; and the entropy variation is zero on average, independently of $t$ the number of steps between the two observations:
 \begin{align} \sum_{x\in\Sigma} \left(S(f^{t}(x))-S(x)\right) = \sum_{x\in\Sigma} S(x)-\sum_{x\in\Sigma} S(x) = 0.\end{align}
In order to witness an arrow of time, we must start from a low entropy configuration, but in practice the equivalence relation and therefore the entropy function are chosen so that low-entropy configurations are non-generic.

{\em Past hypothesis.}
So, to this day, phenomena such as the dilution of the drop of dye in a glass of water, and the increasing entropy therein, are paradigmatic of current understanding of the arrow of time problem\ldots and yet, a careful inspection of the assumptions underlying the conventional argument shows that it only displaces the problem. The question ``Why do we observe an arrow of time'' has become ``Why was the Universe originally of low entropy?''. In fact the conventional argument requires that the entropy at the Big Bang be so low that an arrow of time is still observable $\sim 13.7$ billion years later---making it a very strongly non-generic configuration. This strong assumption of a low entropy initial configuration is referred to as the `past hypothesis' \cite{AlbertPastHypothesis}, and was criticised right from its birth, on account of this unlikelihood \cite{Boltzmann}. Luckily, more recent accounts of the arrow of time suggest we could do without it \cite{CarrollChen,CarrollChen2,BarbourNBody,BarbourJanus}. The key ingredient is expansion.

\section{Time arrow without past hypothesis proof}
\label{sec:main}

In this section we prove via a toy model that an entropic arrow of time can originate from expansion, and that this expansion can be implemented locally and reversibly. In this model Remark~\ref{periodic_entropy} does not apply because, although each configuration is finite, the state space itself is infinite, as configurations can grow. That Remark~\ref{periodic_entropy} can be circumvented in an infinite state space is not surprising by itself: think of $\N$ the set of integers for instance. Each number can be written with a finite number of digits, but the set itself is of course infinite, and it is easy to define a non-periodic bijection on the set:
\begin{align} 
f(n)=
\begin{cases}
0 \text{ if } n=1\\
n+2  \text{ if } n \text{ is even }\\
n-2  \text{ else }
\end{cases}
\end{align}
What is much less obvious and harder to prove is the existence of a time-reversible, causal, homogeneous, ultimately expanding dynamics on generic configurations. We do not know of another such model. %

\subsection{State space}
The states of the model have one dimension, these are circular graphs. These circles are of unbounded but finite size, i.e. the infinite line is not allowed. The vertices are equipped with ports $a$ and $b$, and edges go from port to port, each being used exactly once, as in Fig.~\ref{circular_graph}.\\
Each vertex carries an internal state, amongst four possible states: `containing no particle', `containing a particle moving along port $a$', `containing a particle moving along port $b$', `containing two particles'. Notice that this set of internal states is the same as that used to model electrons in gas-on-grid methods \cite{hardyTimeEvolutionTwoDimensional1973}, or in quantum walks to represent the spin of a fermionic particle such as the electron \cite{ArrighiOverview}. Later we will generalize this by associating, to each port of each node, not just one information bit, but two or three.\\

There is one subtlety: vertices are named, and these names form a little algebra. This is so that a vertex $u$ may be able to split into $u.l$ and $u.r$, and later merge back into $u.l \vee u.r$, and that this be in fact the same as just $u$. There is no escaping this formalism in order to achieve both reversibility and local vertex creation/destruction \cite{ArrighiCreation}, particularly if one wants to preserve causality in the quantum regime \cite{amelia_quantum_graphs}. In order to remain self-contained, the full definition of these named graphs is provided in  Appendix~\ref{sec:namedgraphs}.\\
We denote by $\mathcal{C}_n$ the set of circular graphs with $2n$ information bits per vertex ($n$ per port).  For any vertex $x$, we denote  by $p_i(x)$ the value of the $i$-th bit associated to the port $p$ of vertex $x$. 

\subsection{The toy model}

Our main model ($\hs \inter$) is inspired by the Hasslacher-Meyer dynamics \cite{MeyerLGA}. It acts on the set $\mathcal{C}_1$ and consists in composing two steps: $\inter$, then $\hs$. Each of them is a reversible causal graph dynamics, thus so is their composition---the reader can refer to \cite{ArrighiRCGD,ArrighiCreation} for further theoretical aspects about reversible causal graph dynamics, including general definitions. Thus the whole dynamics is a time-reversible. It consists in a composition of steps which, taken individually, are time-symmetric \cite{GajardoTSCA}:

\begin{definition}
A causal graph dynamics $f$ is said to be \emph{time-symmetric} if and only if there exists a causal graph dynamics $T$ such that $T^2=\identity$ and $TfT = f^{-1}$.
\end{definition}

{\em Step $\hs $.}
Let $\shift$ be the operation that moves all particles along their corresponding port. The operation $\hs$ is such $\hs  \circ \hs  = \shift$: it moves particles by half an edge instead, see Fig.~\ref{sqrtS_intuitive}. One way of thinking about this operation is as inverting the roles of edges and nodes in the sense of taking the dual graph. Notice that alternatively, we could have used a (renaming-equivalent but less symmetrical) operation that moves only those particles associated to port $b$, but by a whole edge.

\begin{definition}[$\hs $]
Step $\hs $ is defined for any graph $\gcirc \in \mathcal{C}_n$ as follows:
\begin{itemize}
\item $V(\hs (\gcirc))=\{u.r\vee v.l\mid\{u:b,v:a\} \in E(\gcirc)\}$
\item $E(\hs (\gcirc))=\{\{x':b,y':a\}\mid u,x,y \in V(\gcirc) \text{ and } x'=u.r \vee x.l  \text{ and } y'= y.r \vee u.l\}$
\item $\forall x' \in V(\hs (\gcirc))$, and for all $i \in [1,n], a_i(x)=a_i(u)$ and $b_i(x)=b_i(v)$, where $u$ and $v$ are the vertices of $\gcirc$ such that $x'.l=u.r$ and $x'.r=v.l$ respectively.
\end{itemize}
\end{definition}

This step is both time-reversible and time-symmetrical, since
\begin{align}  \hs  ^{-1} = T \hs  T\end{align}
with $T$ the function exchanging the left and right information bits of each vertex.

{\em Step $\inter$.} 
Step $\inter$ consists in splitting a vertex into two when it holds two particles, or conversely  merging two vertices holding a pair of back-to-back particles, as in Fig.~\ref{fig:intuitive_I}. Thus, a vertex $u$ such as $a_{1}(u)=b_{1}(u)=1$ will produce two vertices $u.l$ and $u.r$ with $a_{1}(u.l)=b_{1}(u.r)=1$ and $b_{1}(u.l)=a_{1}(u.r)=0$. Conversely, two vertices $u$ and $v$, with $a_{1}(u)=b_{1}(v)=1$ and $b_{1}(u)=a_{1}(v)=0$ will merge into a vertex $u \vee v$ such as $a_{1}(u \vee v)=b_{1}(u \vee v)=1$.

This step is obviously time-reversible and time-symmetric, since it is involutive: $\inter^2=\textit{Id}$.

{\em Spacetime diagram.} The evolution of a configuration under toy model $\hs \inter$ can be represented in the form of a spacetime diagrams, e.g. Fig.~\ref{fig:diagram-espace-simple} represents $\hs\inter$. In these diagrams, the spatial dimension is represented horizontally, and dynamical clock time is represented vertically downwards. Each vertex is represented by a cell, separated from its neighbours with a vertical black line. Its internal state of a cell is captured by its colour. The cells are depicted in variable-sizes, allowing each split/merge to be done ``on the spot''.

\begin{figure}[b]\centering
\centering
\includegraphics[width=0.95\columnwidth]{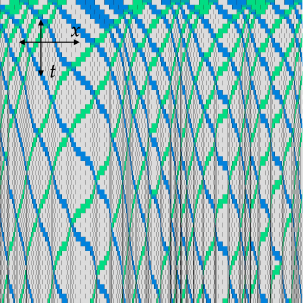}
\caption{\label{fig:diagram-espace-simple}\emph{Spacetime diagram of dynamics $\hs\inter$.} \corrA{Dynamical clock time flows towards the bottom.} Particles moving along port $a$ (resp. $b$) are represented green (resp. blue). \corrA{When they cross, $\inter$ gets applied, leading to shrinking or expansion of the graph. Expansion wins: cells are smaller at the bottom, their packing leading to this ``curtain folds'' appearance.} %
}
\end{figure}

\subsection{Size increases}

The first observation that can be made from Fig.~\ref{fig:diagram-espace-simple} is that the dynamics $\hs \inter$, although time-reversible, grows the size of the graph\corrA{---in the sense of spatial expansion à la general relativity}. It does not grow from the borders, there are no borders, \corrA{space} expands locally. \corrA{It does not grow the number of particles either: those are conserved.} The numerics in Fig.~\ref{fig:evolution_typique_taille_si} suggest this is typical. We will now prove that this happens for almost all initial states. More precisely, we will prove that graphs always end up growing, and that the growth is strict as soon as they contain at least one particle of each type.

\begin{figure}[h]
\centering
\begin{subfigure}{0.95\columnwidth}
\centering
\includegraphics[width=0.95\columnwidth]{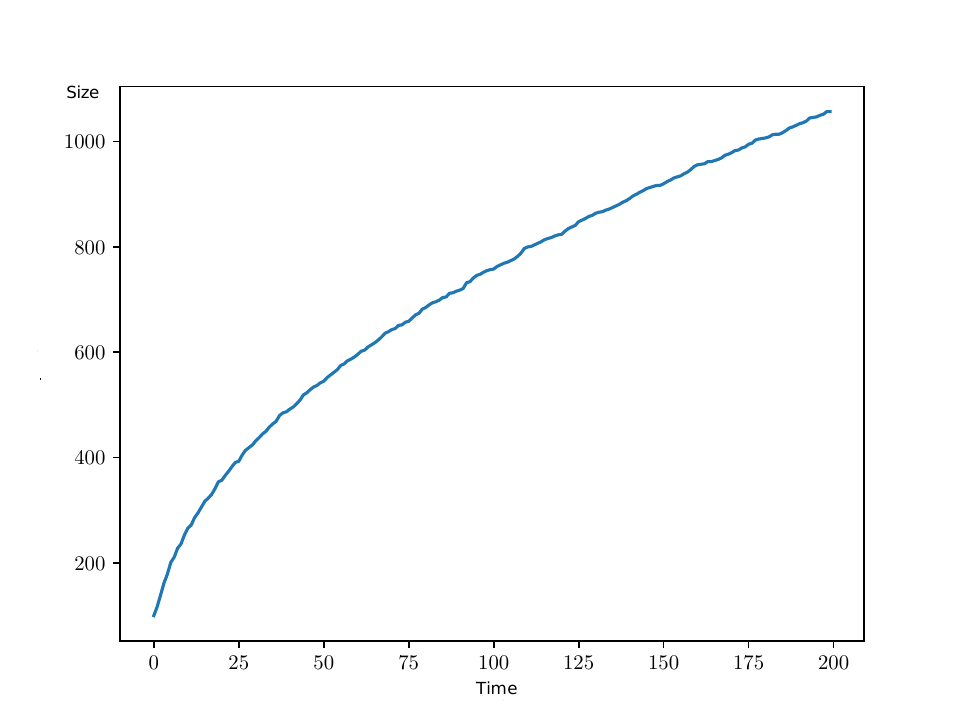}
\caption{\emph{Size of graphs.}}
\label{fig:evolution_typique_taille_si}
\end{subfigure}\\
\begin{subfigure}{0.95\columnwidth}
\centering
\includegraphics[width=0.95\columnwidth]{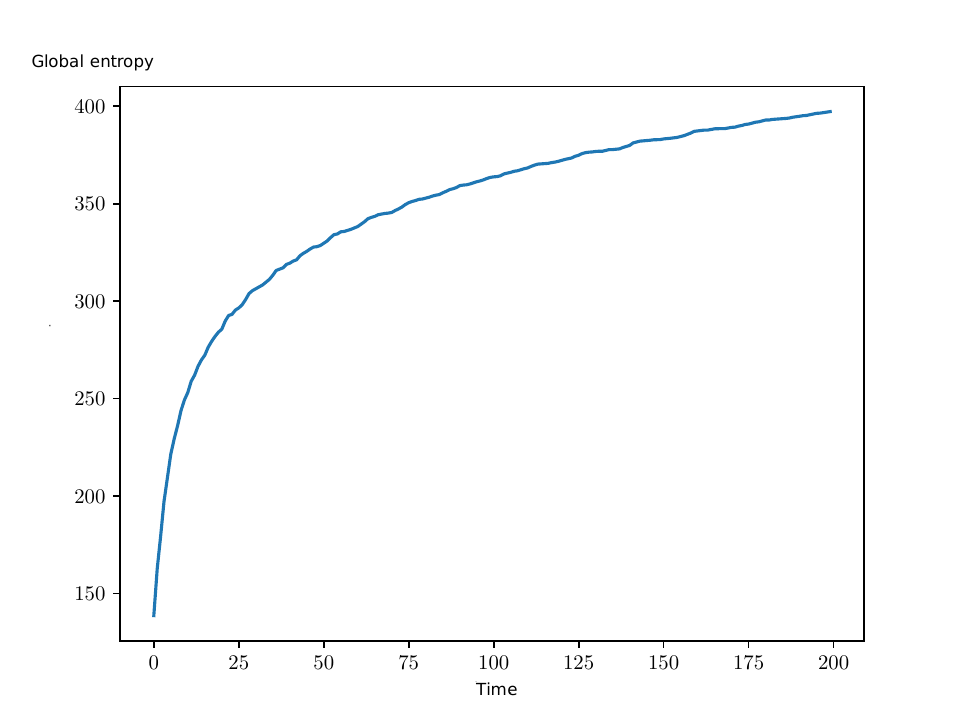}
\caption{\emph{Global entropy of graphs.}}
\label{fig:global_entropy}
\end{subfigure}
\hfill
\caption{Typical size and entropy curves for dynamics $\hs\inter$. The horizontal axis represents the number of steps of the dynamics, aka dynamical clock time.
The initial configuration is drawn uniformly at random amongst all graphs of size $100$.}
\end{figure}

Intuitively this is due to the fact that vertex merger only occurs in the presence of a pattern which is unstable:

\begin{lemma}[Merger Instability]\label{lem:fusion_instability}
Let $\gcirc \in \mathcal{C}_1$ be a circular graph. Given a pair $u$ and $v$ of adjacent vertices of $\gcirc$, these are said to belong to a merger pattern if and only if $a_1(u)=b_1(v)=1$ and $b_1(u)=a_1(v)=0$. For any $u,v$ forming a merger pattern in $X$, there are two vertices $u.r, v.l \in V((\hs \inter)^{-1}(\gcirc))$ such that $u.r$ and $v.l$ form a merger pattern.
\end{lemma}
\begin{proof}
By inspection of Fig.~\ref{fig:proof_instability}, which represents the pre-image of a merger pattern.
\end{proof}

Let us fix notations before we state the expansion theorem. First, in what follows we will write $(u_n)$ to designate the sequence $(u_n)_{n\in \N}$ in the absence of ambiguity. Second, we will use the following asymptotic notation:
\begin{definition}
Let there be two sequences $(u_n)$ and $(v_n)$. We say that $(u_n)$ is of the order of $(v_n)$, and write $(u_n) = \Theta ((v_n))$ if there exist positive numbers $a, b \in \R$ and $n_0 \in \N$ such that for all $n \geq n_0$, $a.v_n \leq u_n \leq b.v_n$.
\end{definition}

\begin{theorem}[Expansion]\label{th:sqrtSI_growing}
For any $\gcirc\in \mathcal{C}_1$ containing at least one particle of each type, let $u_n=|V\left((\hs \inter)^n(\gcirc)\right)|$. We have $\left(u_n\right)=\Theta ((\sqrt{n}))$.
\end{theorem}

\begin{proof}(Outline). As proved in Lem.~\ref{lem:fusion_instability} dynamics $\hs\inter$ cannot create a merger pattern as they are stable by $(\hs \inter)^{-1}$. This entails that any interference disrupting a merger pattern will permanently destroy it. We then prove, by means of a strictly decreasing measure, that such interference will occur if and only if the graph contains at least two particles going in opposite direction. Lastly we  quantify the growth rate once all merger patterns have been removed from the graph, placing bounds corresponding to the best and worst case scenarios. The proof technique is thus akin to a program termination and complexity analysis proof. For readability, the technical details are given in Appendix~\ref{sec:expansionproof}.
\end{proof}

As can be seen in Fig. ~\ref{fig:evolution_typique_taille_si}, for a randomly chosen configuration, the asymptotic regime is reached quickly and is quite stable.

\begin{figure}%
\centering\includegraphics[scale=2]{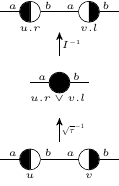}
\caption{\emph{Instability of merger patterns under $\hs\inter$.}}\label{fig:proof_instability}
\end{figure}
\newcommand{\entropyp}{S'}
\subsection{Entropy increases with size}
It turns out that growth in the size of the graph implies growth of entropy, for a natural notion of entropy.\\
Indeed from this point on, we will focus on the entropy function associated with the following equivalence class: two configurations are considered equivalent if and only if they have the same size and the same number of particles. This entropy function can be seen as analogous to the one used in the study of perfect gases.\\
This is because the entropy of a perfect gas configuration is generally associated with the macroscopic properties of pressure, volume, temperature and number of particles. As these four variables are related by the perfect gas law, three are independent. Moreover, in our toy model the speed of the particles is constant, which makes it unnecessary to consider the temperature. We therefore have to consider just two variables amongst: the number of particles; the size of the graph (analogous to the volume); and the density of particles (analogous to the pressure).

We ignore the names of the vertices when counting the number of microscopic configurations corresponding to each macroscopic state, i.e. when counting the number of graphs of a given size and having a given number of particles. We denote $\comb{p}{n}$ the binomial coefficient $p$ among $n$, i.e. $\comb{p}{n} = \frac{n!}{p!(n-p)!}$.

\begin{definition}\label{def:entropy}
The entropy function $\entropyp$ is defined by $\entropyp(\gcirc) = \log(\comb{p}{2|V(\gcirc)|}$ where $p$ is the number of particles in $\gcirc$.
\end{definition}

One could also have chosen to refine the macroscopic states by separating the two type of particles, but this will lead to the same behaviours.
In the case of $\hs \inter$ we have proven in Th.~\ref{th:sqrtSI_growing} that $V(\gcirc)$ grows as a square root. Because this rule preserves the number of particles we automatically obtains the growth of entropy. But we can be more precise:

\newcommand{\orderof}{\Theta}
\begin{corollaire}\label{cor:growing_global_entropy}
For any $\gcirc\in \mathcal{C}_1$ containing at least one particle of each type, let $e_n=\entropyp(\hs \inter)^n(\gcirc)$. We have $(e_n)=\Theta((\log(n)))$.
\end{corollaire}

\begin{proof}
Thanks to Th.~\ref{th:sqrtSI_growing}, we have that $(\entropyp(\hs \inter)^n(\gcirc))_{n\in \N} = \orderof(\log(\comb{p}{\orderof(\sqrt n)}))$.
Using the bounds $(\frac{a}{b})^b \leq \comb{b}{a} \leq  e^b(\frac{a}{b})^b$ \cite{dasBriefNoteEstimates2015},
and noting that $p$ is constant, we obtain :
\begin{align}  (\entropyp(\hs \inter)^n(\gcirc))_{n\in \N} &= \orderof(p \log(\frac{\orderof(\sqrt{n})}{p})) \\ &= \orderof(\log(n^{1/2})) \\ &= \orderof({\log(n)})\end{align}
\end{proof}

Note that this corollary would apply equally well to any dynamics where the number of particles remains constant and the size of the graph grows polynomially (not necessarily as a square root).

The asymptotic regime of the size function was reached as soon as all the merger patterns were destroyed, cf. Fig.~\ref{fig:evolution_typique_taille_si}. The same happens with global entropy, cf. Fig.~\ref{fig:global_entropy}.

If we are only interested in whether entropy grows, without seeking to characterise its asymptotic behaviour, we can state a more general theorem, relating it to size growth. Indeed, under the assumption that the particles do not fill the whole space, nor disappear completely, then any dynamics that increase the size of the graph will also increase the entropy:

\newcommand{\limite}[3]{\lim\limits_{#1} #2 = #3}
\begin{theorem}[Entropy increases with size]\label{th:entropywithsize}
For all $\gcirc\in \mathcal{C}_u$ and $f : \mathcal{C}_u \to\mathcal{C}_u$ such that :
 \medskip
\begin{itemize}
\item $\limite{n \to+\infty}{|V(f^n(\gcirc))|}{+\infty}$
\item $\exists m \in \N$ such that $\forall n \geq m, 1 \leq p_n \leq  2u \times|V(f^n(\gcirc))| -1$ where $p_n$ is the number of particles in the step in $f^n(\gcirc)$.
\end{itemize}
\medskip
We have that $\limite{n \to+\infty}{\entropyp(f^n(\gcirc))}{+\infty}$.
\end{theorem}

\begin{proof}
Thanks to the second condition, we have for all $n\geq m$ :
\begin{align} \entropyp(f^n(\gcirc)) &= \log(\comb{p}{n}_{2|V(f^n(\gcirc))|})\\  
&\geq \log(\comb{1}{2|V(f^n(\gcirc))|})\\
&= \log(2|V(f^n(\gcirc))|)
\end{align}
As $\log(2|V(f^n(\gcirc))|)$ tends to $+\infty$ when $n$ tends to $+\infty$, this is also the case for $\entropyp(f^n(\gcirc))$.
\end{proof}

\subsection{Recovering an arrow of time}

With Th.~\ref{th:sqrtSI_growing} and Cor.~\ref{cor:growing_global_entropy}, we have proven that an entropic arrow of time emerges in some time-reversible, causal, homogeneous laws (namely the $\hs \inter$ toy model), without relying on the past hypothesis. More precisely, we have proven that starting from almost all configurations, entropy ultimately grows as we iterate the dynamics. Intuitively, after a finite period of dynamical clock time, the entropic clock's arrow aligns with that of the dynamical clock. This solves criticism $(iii)$ of the conventional argument. Notice how, ultimately, this resolution boils down to the fact that configurations are of finite but unbounded size. In this context, assuming that the universe ``starts small'' is reasonable, because for any configuration, there are many more larger configurations than smaller ones. The same happens with entropy: any starting value is small within the set of positive real numbers. In that sense past low entropy is no longer unreasonable, it is unavoidable.

An immediate consequence is that the toy model is not periodic, i.e. there is no recurrence time: this solves criticism $(i)$ in a way more satisfactory manner than arguing that ``there is a recurrence time but it is typically too big to be observed''. Let us look at criticism $(ii)$.

Since the system $\hs \inter$ is time-reversible, one can naturally ask what happens if one tries to "go back in time", i.e. how a graph evolves when one applies the dynamics $(\hs \inter)^{-1}=\inter^{-1}\hs^{-1}$.
Numerics suggest it also increases the size of the graph, but at a different rate, see Fig.~\ref{time_size_sh_rev}. We can prove it:

\begin{theorem}\label{sqrtSI_rev_growing}
For any $\gcirc\in \mathcal{C}_1$ containing at least one particle of each type, the sequence $(|V((\hs \inter)^{-n}(\gcirc))|)_{n\in \N} = \Theta(n)$.
\end{theorem}
\begin{proof}
In the absence of patterns \lineoftwo{black}{white}{white}{black}, the size of the graph decreases strictly each time two particles meet. By conservation of momentum, the particles will continue to cross each other. Since the graph cannot decrease continuously a pattern \lineoftwo{black}{white}{white}{black} will inevitably form. As can be seen in the proof of Lem.~\ref{lem:fusion_instability}, the pattern \lineoftwo{black}{white}{white}{black} is stable by $\inter^{-1}\hs^{-1}$, and cannot be crossed by other particles. This implies that once such a pattern is present, any pair of particles not belonging to such a pattern can only collide once. When all these collisions have occurred, each application of $\inter^{-1}\hs^{-1}$ increases the size of the graph by the number of patterns present. We can bound by $\min(n_a,n_b)$ the number of such patterns in $\gcirc$ and its successors by $\inter^{-1}\hs^{-1}$. Thus there exists $m \in \N$ such that for all $n \geq m$ we have :
\begin{align} &|V((\hs \inter)^{-m}(\gcirc))|+ n \\
&\leq |V((\hs \inter)^{-n}(\gcirc))|\\
&\leq |V((\hs \inter)^{-m}(\gcirc))| + n\min(n_a,n_b)\end{align}
\end{proof}

\begin{figure}
\centering\hfill
\begin{subfigure}{0.95\columnwidth}\centering
\includegraphics[width=0.95\columnwidth]{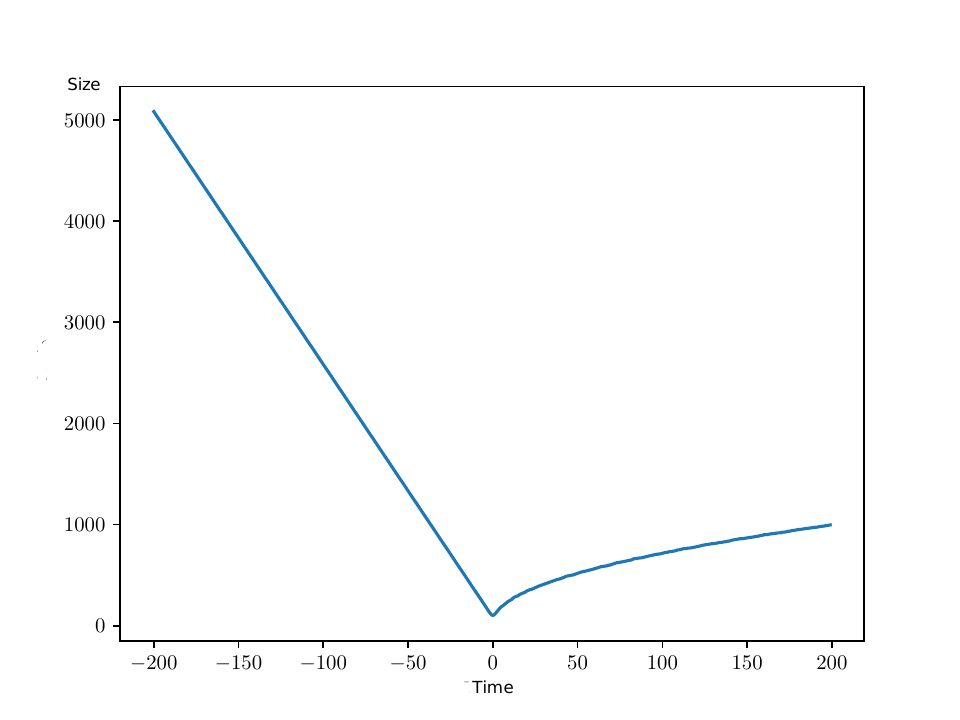}
\caption{\emph{Typical curve of size of the graphs under $(\hs \inter)^{-1}$ and $\hs \inter$.} The initial condition is drawn uniformly at random amongst graphs of size $100$. We run the dynamics backwards and forward to explore negative and positive dynamical clock times. }
\label{time_size_sh_rev}
\end{subfigure}\\
\begin{subfigure}{0.95\columnwidth}\centering
\includegraphics[width=0.95\columnwidth]{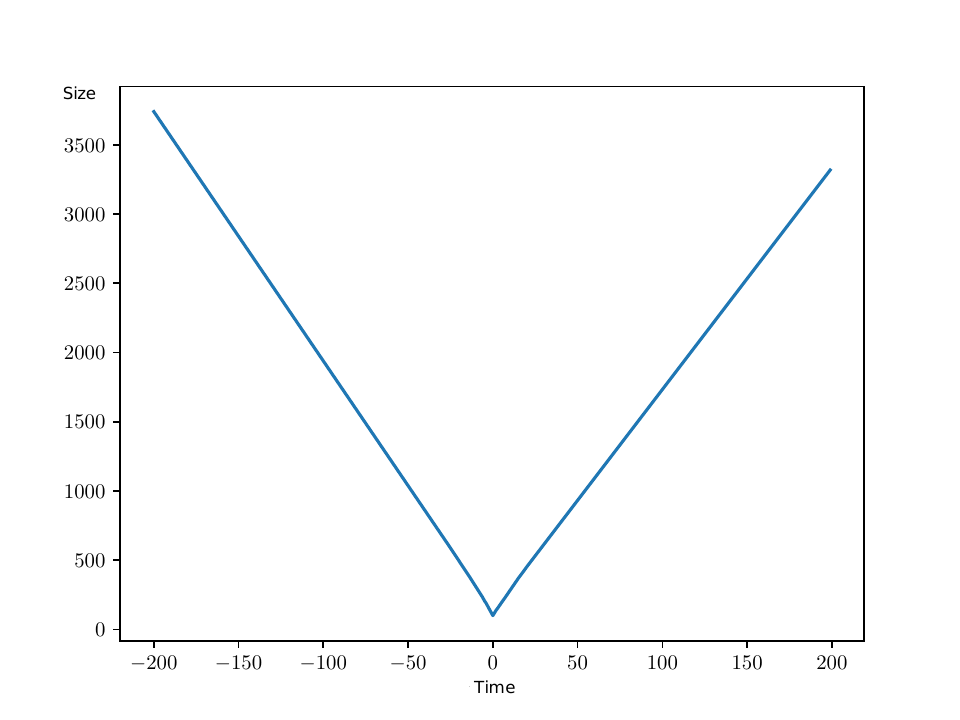}
\caption{\emph{Typical curve of size of the graphs under $\hs\interb$}, a time-symmetrised version of $\hs\inter$. The initial condition is drawn uniformly at random amongst graphs of size $100$, thus particle density is high $\sim 0.5$.}
\label{sh_etendu}
\end{subfigure}\hfill
\caption{Evolution of the size of a typical graph under $(\hs \inter)^{-1}$, $\hs \inter$ and $\hs\interb$.}
\end{figure}

By the same proof scheme as for Corollary~\ref{cor:growing_global_entropy}, we obtain the entropy growth for $\inter^{-1}\hs^{-1}$:
\begin{corollaire}
For any $\gcirc\in \mathcal{C}_1$ containing at least one particle of each type, let $e_n=\entropyp(\hs \inter)^{-n}(\gcirc)$. We have $(e_n)=\Theta(\log(n))$.
\end{corollaire}

Thus, a variation of criticism $(ii)$ still holds, as entropy increases in both directions from a `source'. But here, since the model is not periodic, there is a single such source: a possibility discussed under the name of `Janus point', \corrB{whether positively  \cite{CarrollChen,CarrollChen2,BarbourNBody,BarbourJanus} or negatively \cite{ZehAgainstBarbour}---including on conceptual grounds}.
To move away from this region of minimal entropy, whether by iterating the forward dynamics or its reverse, means augmenting entropy and hence `going towards the future'. The important point is the way in which this variation relates to the resolution of criticism $(iii)$: it is no longer necessary to assume that the universe ``started'' in an improbable low-entropy state; as the existence of an arrow of time just follows from the existence of a minimal region, which itself is a direct consequence of the dynamical law.

\section{Time-symmetric model}\label{sec:time-symmetry}

Although the $\hs \inter$ toy model is time-reversible, it is not time-symmetric:
\begin{remarque}
 If $f$ is time-symmetric by $T$ then for all $n \in \N, Tf^nT= f^{-n}$. Consequently, since $\hs\inter$ and $\inter^{-1}\hs^{-1}$ do not have the same asymptotic behaviour, they are not time-symmetric.
\end{remarque}

However, there is a way to time-symmetrize the dynamics, by extending the state space to $2$ bits of information per port, i.e. by working on $ \mathcal{C}_2$. First of all notice that step $\hs$ is already defined for an arbitrary $ \mathcal{C}_n$, and time-symmetric under conjugation by the involution $T$. We only need to extend step $\inter$, because although $\inter$ is not time-symmetric under conjugation by that same conjugation $T$. Indeed, the orientation of the particles affects the behaviour of $\inter$: a pattern may merge or not depending on whether we apply $T$. %

We thus time-symmetrize step $\inter$ into step $\interb$, adding $2$ bits of information, for which step $\interb$ will trigger split/mergers, not on \lineoftwo{black}{white}{white}{black}, but on \lineoftwo{white}{black}{black}{white}, as in Fig.~\ref{intuitive_super_I}. In other words, step $\interb$ acts on the second bit, just like $T \inter T$ would act on the first bit.

\begin{figure}
\begin{center}
\begin{tabular}{l r}
\includegraphics[scale=2, clip=true, trim=0 0.0cm 0 0]{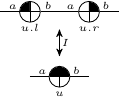} &
\includegraphics[scale=2, clip=true, trim=0 0.0cm 0 0]{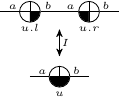}
\end{tabular}
\end{center}
\caption{\emph{Dynamics $\interb$}. The presence of a particle is represented in black. Now there are two types of particles corresponding to port $a$, and two corresponding to port $b$. \corrA{This is what allows us to restore the time-symmetry that was broken by rule $\inter$.}}
\label{intuitive_super_I}
\end{figure}

\begin{proposition}
The dynamics $\hs\interb$ is time-symmetric.
\end{proposition}

\newcommand{\flip}{R}
\begin{proof}
We define the function $\flip : \mathcal{C}_2 \to \mathcal{C}_2$ which for any vertex $u$ exchanges the values of $a_{1}(u)$ and $b_{2}(u)$, as well as the values of $a_{2}(u)$ and $b_{1}(u)$. One can easily verify that such a dynamic commutes with $\interb$, i.e. that $\interb\flip = \flip\interb$. Moreover, since $R$ exchanges particles on the ports $a$ and $b$, we have that $\flip\hs\flip=\hs^{-1}$. By letting $T:=\interb\flip$, we obtain the equalities:
\begin{align} T^2 = \interb\flip\interb\flip = \interb\interb\flip\flip = \identity\end{align}
And also :
\begin{align} T\hs\interb T &= \interb\flip\hs\interb\interb\flip \\
&=  \interb\flip\hs\flip \\
&= \interb\hs^{-1}\\
&=\interb^{-1}\hs^{-1} 
\end{align}
\end{proof}

The symmetry comes at a price: with potentially $2$ types of particles per node, interactions through step $\interb$ are sometimes prevented in a way that makes it difficult to extend Th.~\ref{th:sqrtSI_growing}. Still, we have strong numerical evidence that generic graphs expand, with a size curve of order $\orderof(n)$, see Fig.~\ref{sh_etendu}.

Just like for $\hs \inter$, the system has a minimal configuration from which two time arrows pointing in opposite directions are derived. In our experiments, this configuration is often the initial configuration. This does not mean that it is non-generic, quite the contrary: by drawing a configuration uniformly at random from the set of graphs of a certain size, we obtain with high probability a density that is close to $2$ particles per node. On the other hand, if we were to randomly choose a configuration $\gcirc$ from the set of configurations of an orbit, then for all $\epsilon \in \left[0,1\right]$, the density of $\gcirc$ would be less than $\epsilon$ with probability $1$ (if it is indeed true that the size of the graph grows in $\orderof(n)$ that is). By lowering the particle density of the initial state, we lower its chances of being the minimal configuration, see Fig.~\ref{fig:low_density}.

\begin{figure}
\includegraphics[width=0.95\columnwidth]{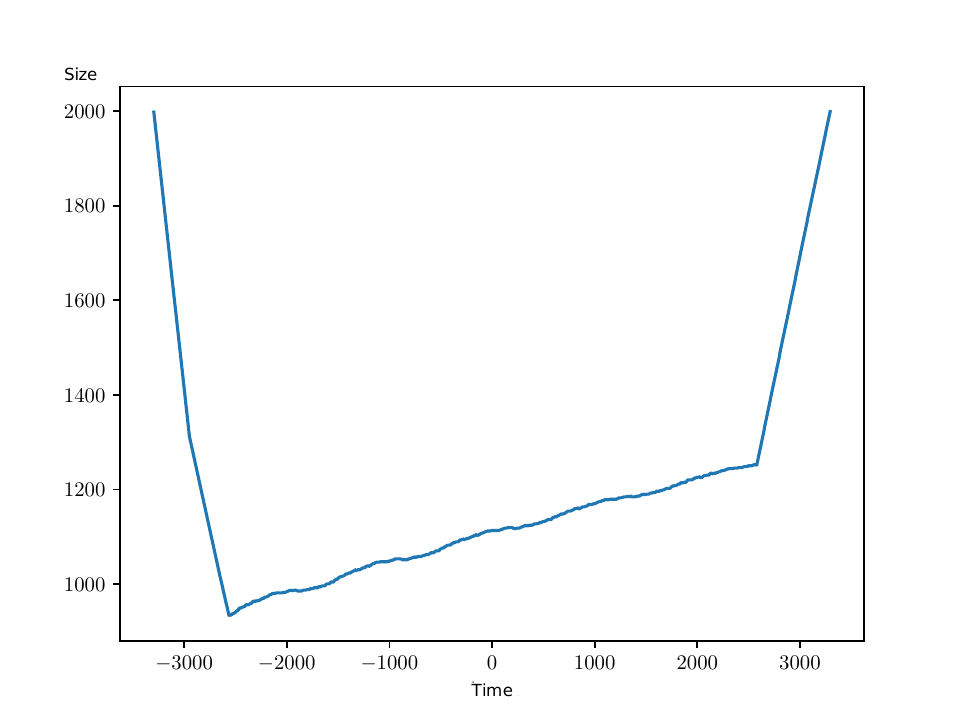}\\
\includegraphics[width=0.95\columnwidth]{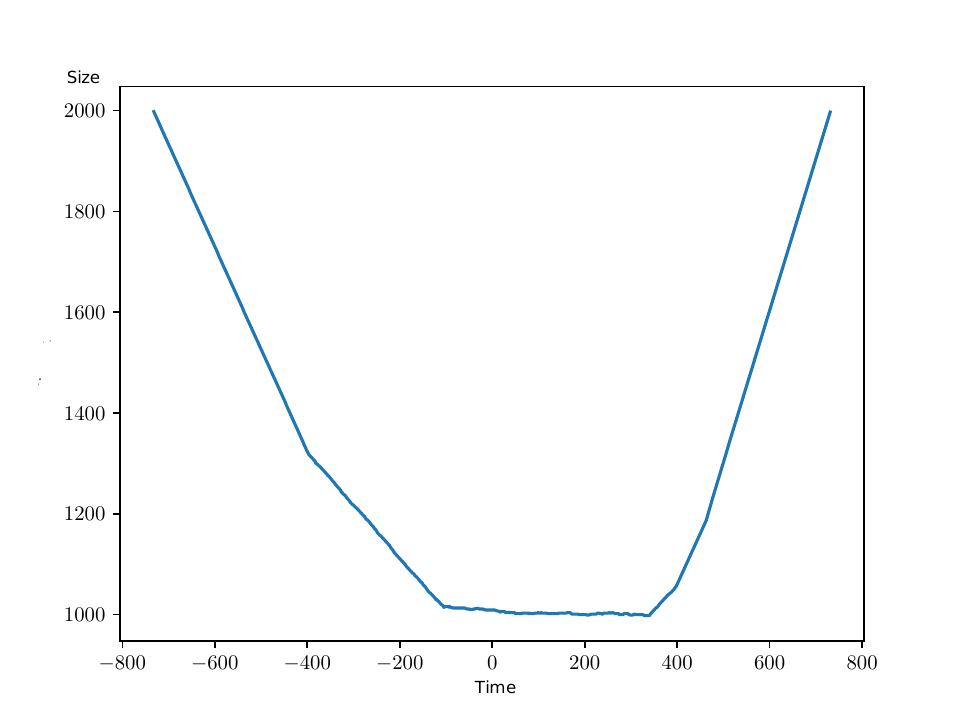}
\caption{\emph{Typical curve of size of the graphs under $\hs\interb$, starting from low density initial conditions}. The initial condition is drawn at random amongst low density graphs of size $1000$, in such a way that each bit be set with probability $0.01$. Time $0$ is chosen a posteriori, so that  size $2000$ be reached. A chaotic phase is followed by linear growth.}
\label{fig:low_density}
\end{figure}

\section{Exponential growth models}\label{sec:exponential}
\newcommand{\interd}{I_e}
\newcommand{\ecp}{F}

The very early universe is believed to have known a phase of exponential growth, called inflation. Is this compatible with reversibility? Intuitively, we can achieve this in a variant of our toy model if the growth of the configuration is made proportional to the size of the configuration. This was not the case with $\hs\inter$ because it preserves the number of particles. This will be fixed if we manage to maintain the particle density above a certain limit. There are two simple ways of doing so. 

The first is to create new particles at each node division, for example by replacing step $\inter$ with the step $\interd$ shown in Fig. ~\ref{fig:intuitive_I_exp}. Under dynamics $\hs \interd$, the particle density tends towards $1/2$, which naturally leads to exponential growth according to numerics, see Fig.~\ref{fig:SH_exp_Taille_100}.

\begin{figure}\centering
\begin{subfigure}{0.95\columnwidth}\centering
\includegraphics[width=0.94\columnwidth]{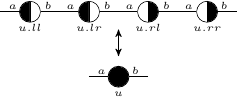}
\caption{\emph{Step $\interd$}. \corrA{All occurrences of the two patterns are flipped for each other synchronously.} Thus, particles crossing on a vertex trigger a split into a pattern of $4$ vertices carrying $2$ new particles. \corrA{This way a certain level of particle density is maintained throughout expansion.}}
\label{fig:intuitive_I_exp}
\end{subfigure}\\
\begin{subfigure}{0.95\columnwidth}\centering
\includegraphics[width=0.65\columnwidth]{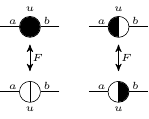}
\caption{\emph{Dynamics $\ecp$.} \corrA{Particle absence is replaced by particle presence and reciprocally, synchronously across the graph.}}
\label{fig:intuitive_B_flip}
\end{subfigure}
\caption{\label{fig:exponentialgrowth} {\em Two dynamics leading to variants featuring exponential growth.}}
\end{figure}
\begin{figure}\centering
\begin{subfigure}{0.95\columnwidth}
\centering
\includegraphics[width=0.95\columnwidth]{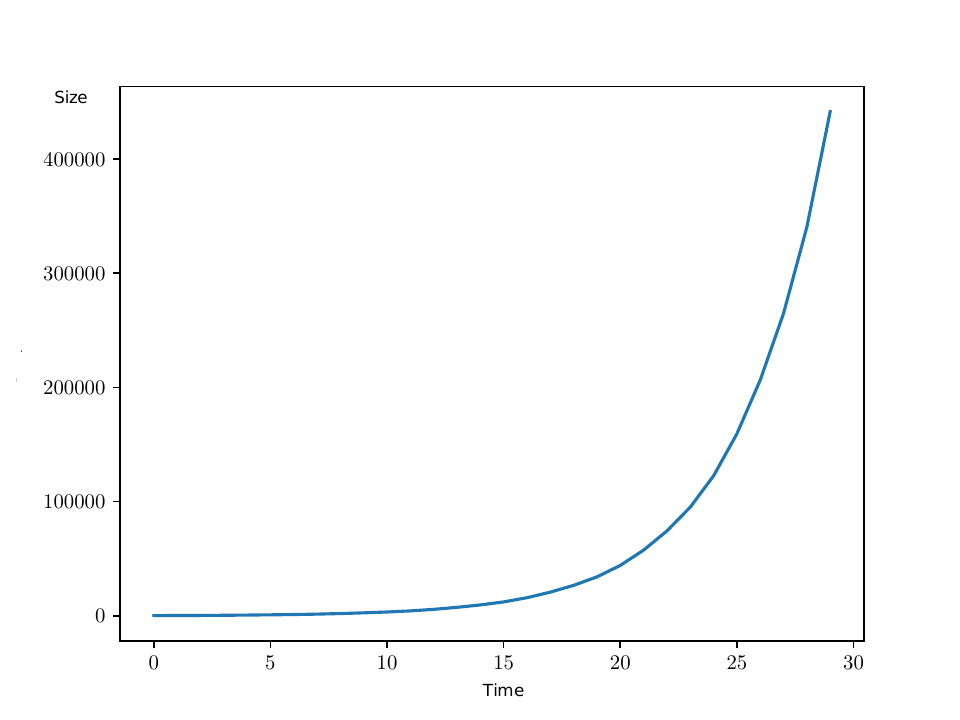}
\caption{\emph{Typical curve of size of the graphs under $\hs\interd$.} The initial condition is drawn uniformly at random amongst graphs of size $100$ and density $\sim 0.1$.}\label{fig:SH_exp_Taille_100}
\end{subfigure}\\
\begin{subfigure}{0.95\columnwidth}
\centering
\includegraphics[width=0.95\columnwidth]{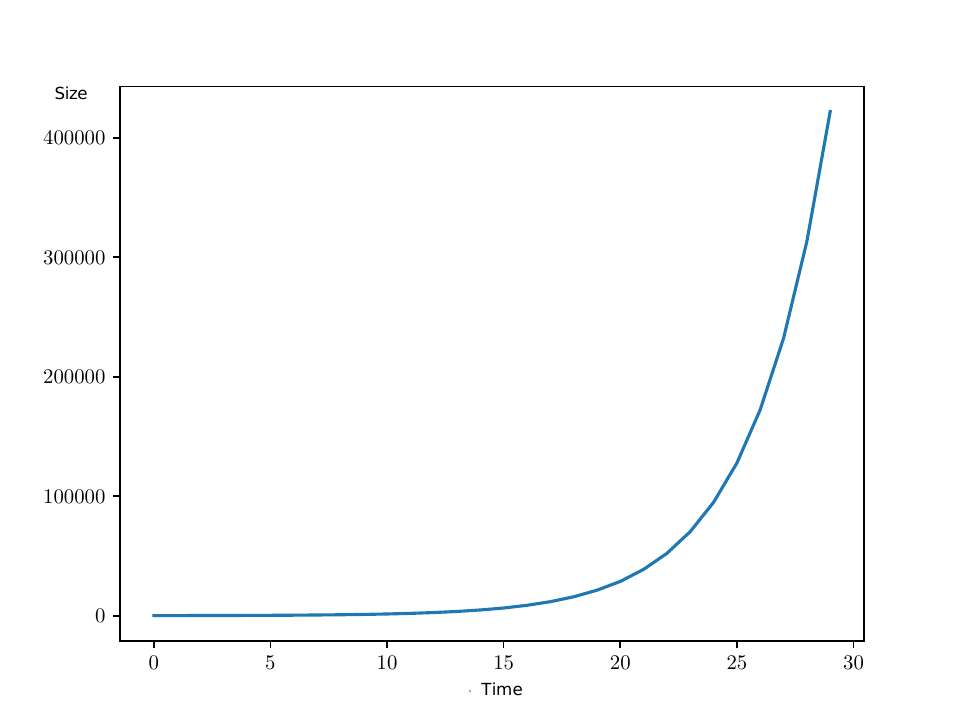}
\caption{\emph{Typical curve of size of the graphs under $\hs\inter\ecp$.} The initial condition is drawn uniformly at random amongst graphs of size $100$.}
\label{fig:SH_exp_flip_Taille_100}
\end{subfigure}
\caption{Observed exponential growth for the variants.\label{fig:observedexponentialgrowth}}
\end{figure}

The other possibility is to flip all the information bits of each vertex at each iteration, via step $\ecp$, as illustrated in Fig. ~\ref{fig:intuitive_B_flip}. Dynamics $\ecp\hs\inter$ does not actually keep the particle density above a bound, but it does oscillate between more or less than $1$ particles per vertex. This ensures that at least every second iteration, density is above $\frac{1}{2}$. Again this leads to exponential growth according to numerics, see Fig.~\ref{fig:SH_exp_flip_Taille_100}.

\corrB{The point of these two new variations of the toy model is not so much to reproduce specific features of contemporary cosmology. They are but proof of concepts. Still, they serve to show that 1/ there is no a priori incompatibility between exponential, inflationary growth and reversible expansion mechanism and 2/ reversible expansion mechanisms seem a robust recipe for time arrow without past hypothesis, independent of the specifics of each model.}

\section{Fighting thermal death}\label{sec:localentropy}

The previous sections provide a plausible toy model explanation for the entropic clock's arrow of time at global scales. But, is this explanation good enough to also explain the arrow of time that we observe locally, at our scales, e.g. going back to the glass of water with a drop of dye of Sec.~\ref{sec:conventional}? The fact is that local entropy quickly drops close to zero in the toy model, due to the dilution of matter under expansion. This is also the dominant process in cosmology, where it is known as `thermal death' (aka ``big freeze"). This section provides a mechanism for whereby clumps of matter form, slowing down thermal death. That way local entropy still can increase at places, and the arrow of time be witnessed locally. 

\begin{figure}[b]\centering
\centering
\includegraphics[width=0.95\columnwidth]{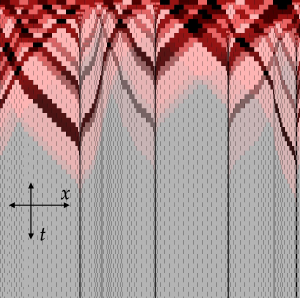}
\caption{\emph{Spacetime diagram and local entropy of dynamics $\hs\interb$.} Particles are shown in levels of grey, local entropy is shown in levels of red. \corrA{As space expands, particles are diluted and local entropy drops. This is thermal death.} \label{diagram-espace-entropie}}
\end{figure}

\subsection{Prognosis}

\paragraph*{Local entropy dies out}
In order to define a notion of local entropy entropy we apply the previous notion to disks of radius $r$ of the configuration. Here we average that:

\begin{definition}[Average local entropy]
Let $S$ be an entropy function on $\X$, and $r \in  \N^+$. The \emph{average local entropy function} of radius $r$ is defined by :
\begin{align} S_r(X) := \frac{1}{|V(X)|}\sum_{u\in V(X)} S(X_u^r)\end{align}
\end{definition}

The average local entropy turns out to be bounded by the density of the graph:

\begin{theorem}
For all graphs $X$ and for all $r \in \N$ :
$S_r(X)\leq r\frac{p(X)}{|V(X)|}$, where $p(X)$ is the number of particles of $X$.
\end{theorem}
\begin{proof}
\begin{align*}
S_r(X) &= \frac{1}{|V(X)|}\sum_{u\in V(X)} S(X_u^r)\\
&= \frac{1}{|V(X)|}\left(\sum_{\substack{u\in V(X)\\ p(X_u^r)>0}} S(X_u^r) + \sum_{\substack{u\in V(X)\\ p(X_u^r)=0}} S(X_u^r)\right)\\
&= \frac{1}{|V(X)|}\sum_{\substack{u\in V(X)\\ p(X_u^r)>0}} S(X_u^r) \\
& \leq \frac{1}{|V(X)|} \times rp(X)
\end{align*}
\end{proof}

If the dynamics grows the size of the graph whilst preserving the number of particles, we cannot escape thermal death:
\begin{corollaire}[Thermal death]\label{cor:thermaldeath}
For any dynamics $F$ that has a bounded number of particles and increasing graph sizes, for all $r \in \N$ : $\limite{n \to +\infty}{S_r(F^n(X))}{0}$
\end{corollaire}

At first glance this result may seem to go against our everyday experience, whereby we witness an arrow of time locally. Yet it is also true that as the Universe expands, it becomes mostly empty, and its average local entropy drops: thermal death is one of the envisaged fates for the Universe.

\begin{figure}[b]%
\begin{center}
\includegraphics[width=0.95\columnwidth]{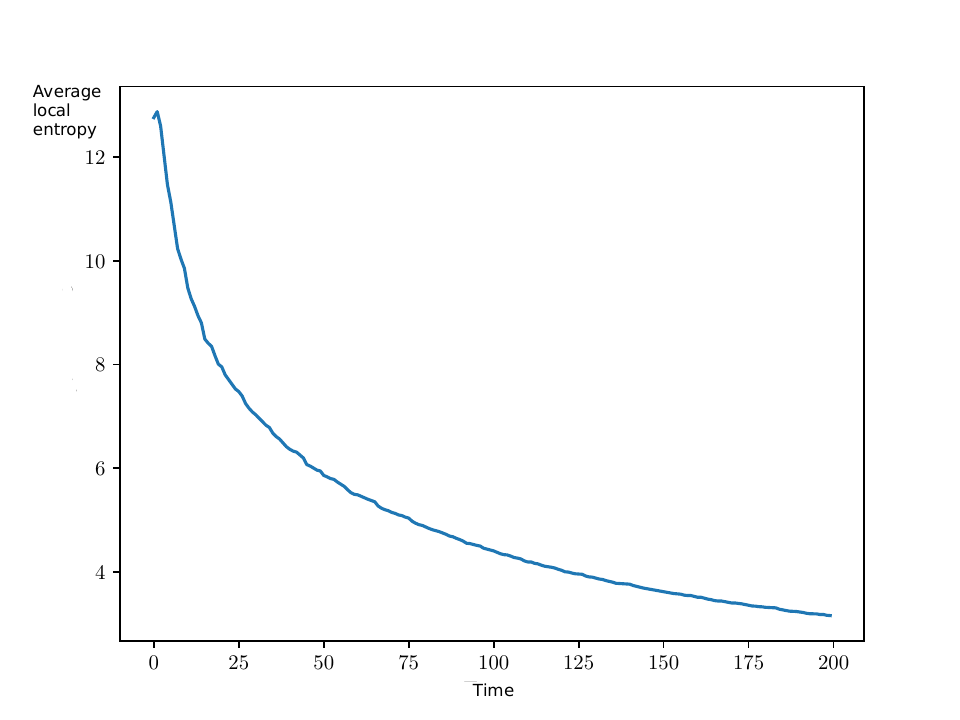}
\end{center}
\caption{\emph{Typical curve of average local entropy of the graphs under $\hs\inter$.} The initial condition is drawn uniformly at random amongst graphs of size $100$.}
\label{avg_entropy}
\end{figure}

Corollary~\ref{cor:thermaldeath} is verified numerically for dynamics $\hs\inter$, see Fig.~\ref{diagram-espace-entropie} and Fig.~\ref{avg_entropy} as well as for the time-symmetrized variant and their reverse dynamics. 
In the case of $(\hs\inter)^{-1}$, we can show that after a certain time, stable patterns are formed that keep on producing new vertices. The other particles cannot cross these patterns, they end up constantly moving in the same direction and thus not crossing each other (dynamics $\hs\interb$ seems to behave likewise).

In order to still witness growth in local entropy, in spite of expansion, average local entropy for fixed-size windows is therefore not the right mathematical quantity. We must look for a quantity which compensates for the loss of particle density. One way to proceed is to study average local entropy for variable-size windows, whose size follow expansion. This is done in Appendix~\ref{sec:variableentropy}. Next, we study the sum of the local entropies for fixed-size windows, instead.

\paragraph*{The sum of local entropies reaches a plateau}
\newcommand{\sommeent}[2]{\overline S_{#1}(#2)}
\newcommand{\ecartent}[2]{\sigma_{#1}(#2)}

\VL{Recall that the local entropy of a fixed-sized window is directly related to the density of particles there. Thus, generally speaking: 1/ Local entropy drops as the density drops and 2/ Once the density is everywhere very low, the local entropy can only increase by gaining particles from the neighbouring window, whose entropy then decreases. I.e. local entropy at this stage is mostly just moving around alongside with the particles. Instead, we wish to have a mathematical quantity that allows us to monitor local entropies in a way that 1/ Factors out the dilution of matter and 2/ Ignores mere displacements from one window to the next.} This all suggests studying the sum of local entropies:

\begin{align} \sommeent{r}{X} =  \sum_{u\in V(X)} S(X_u^r)\end{align}

Unfortunately, the sum of the entropies grows only very briefly for our previous dynamics, before it reaches the plateau of death, as seen in Fig.~\ref{fig:quick_local_entropy}.

\begin{figure}
\includegraphics[width=0.95\columnwidth]{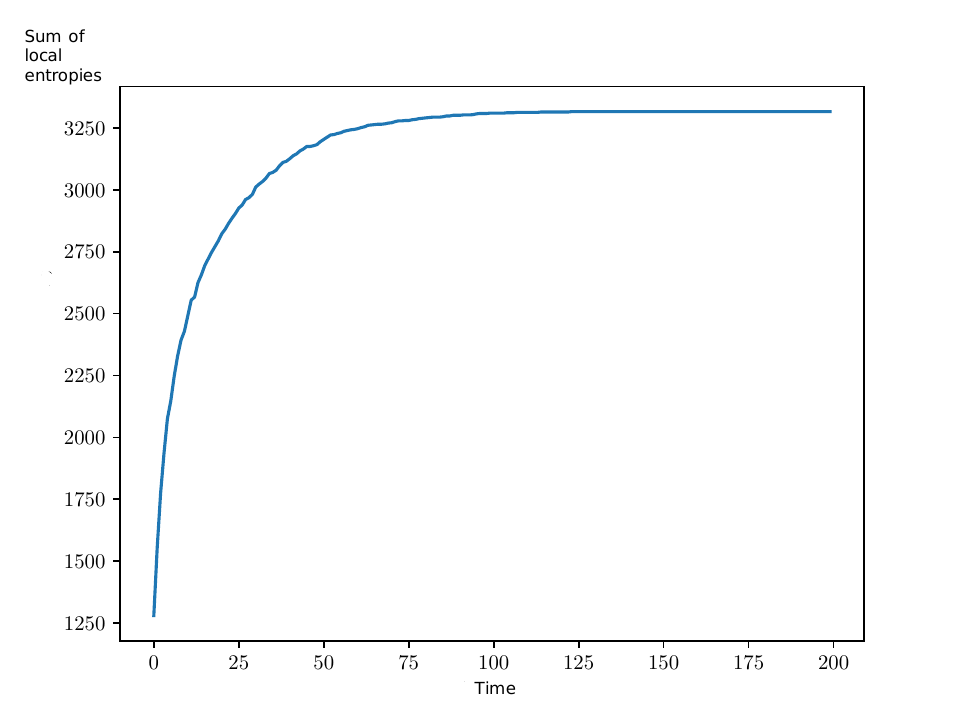}\\
\includegraphics[width=0.95\columnwidth]{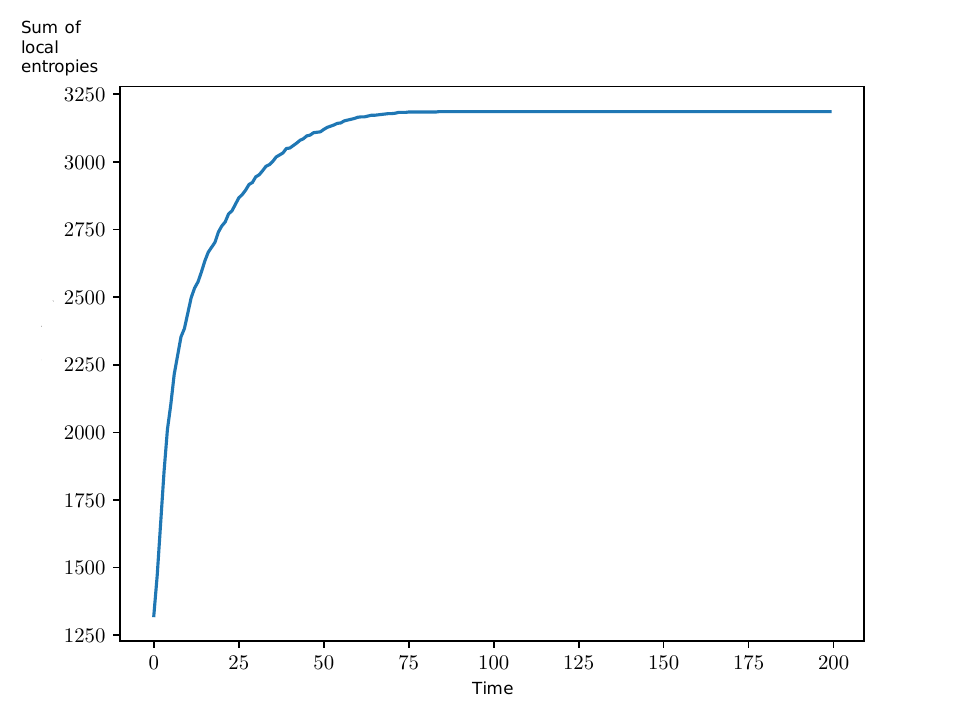}
\caption{\emph{Typical curves of the sum of the local entropies of the graphs under $\hs\inter$ and $\hs\interb$.} The local entropy windows are of radius $5$. The initial conditions are drawn uniformly at random amongst graphs of size $100$.}
\label{fig:quick_local_entropy}
\end{figure}

\VL{\noindent {\em Details per models.} Dynamics $\hs\interb$ and $(\hs\inter)^{-1}$, behave similarly: the growth in local entropy around the time $0$ is the consequence of the growth of the graph. Then, division patterns fragment the space between the particles, and prevent any interaction. Thus, one essentially obtains patterns moving away at constant speed, between which inactive patterns move. This can be easily seen on a spacetime diagram on which the particles and the local entropy are superimposed (see Fig.~\ref{diagram-espace-entropie}). Dynamics $\hs\inter$ has a different behaviour. Observe that the sum of local entropy reaches its maximum when each local window contains at most one particle in each direction. Once all merger patters are destroyed, each time a particle collides with another, the distance between it and the other particles with the same orientation increases. Thus, for local entropy windows of radius $r$, the maximum sum of local entropy will be reached once each particle has completed $2r$ revolutions.
}
 
\subsection{Treatment}
\newcommand{\dampl}{D_l}
\newcommand{\dampr}{D_r}
\newcommand{\damp}{D}
\newcommand{\matter}{matter }
\newcommand{\heat}{radiation }
\newcommand{\agglomeration}{clump }
\begin{figure}%
\begin{center}
\includegraphics[width=0.95\columnwidth]{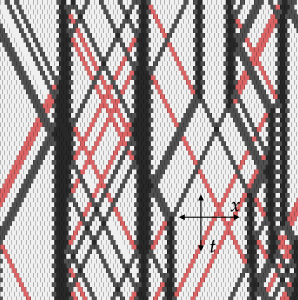}
\end{center}
\caption{\emph{Spacetime diagram of $\hs\damp$.} Matter particles are shown in levels of grey, radiation particles are shown in red. For readability purposes, the initial configuration was taken without radiation. \corrA{Observe how grey signals coalesce upon collision amongst themselves, forming black conglomerates but emitting red signals, which in turn erode the conglomerates that they run into.}}
\label{diagram-espace-SH-matter}
\end{figure}
In order to preserve the possibility of local entropy growth over a longer period of dynamical clock time, we need extra ingredients. There are many avenues that could be explored and nature-inspired models \cite{RovelliBackToReichenbach,RovelliBackToReichenbach2} do not seem entirely out of reach, but they would require modelling efforts that lie beyond the scope of this paper, and whose complexity would likely obscure its main message. Still, some of the key ingredients seem to be matter clumping and radiation, and we can easily start to introduce these by implementing `inelastic shocks'.

\paragraph*{Inelastic shocks}
To do so we use $2$ bits of information per port, i.e. state space ${\cal C}_2$, and consider the composition of dynamics $\hs$ which we already defined, with $\dampl$ and $\dampr$ which we introduce. $\dampr$ is defined to be the permutation the $4$ different patterns shown in Fig.~\ref{intuitive_dampening} (it acts as the identity on all other patterns). $\dampl$ is the left-right-symmetric of $\dampr$, as obtained by 1/ Inverting the position of nodes $u$ and $v$ and their contents 2/ Inside each node, inverting the positions of the particles between ports $a$ and $b$.
\newcommand{\interc}{I_3}
\begin{figure}[b]
\begin{center}
\includegraphics[width=0.95\columnwidth]{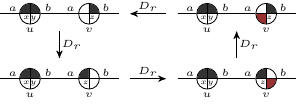}
\end{center}
\caption{\emph{Dynamics $D_r$.} `\matter particles' are shown in black. `\heat particles' are shown in red. The $x$, $y$, $z$ are boolean variables, so that the rule can apply independently of the presence or absence of these particles. Any left-incoming matter particle {\em (bottom-right)} gets stuck emitting radiation {\em (bottom-right)}. From there on the clump is stable: any right-outgoing particle {\em (top-left)} is brought back {\em (bottom-left)}. It will only be freed by left-incoming radiation particles {\em (top-right)}.}
\label{intuitive_dampening}
\end{figure}
Fig. ~\ref{diagram-espace-SH-matter} shows the evolution of a configuration under the dynamics $\hs\damp$. Intuitively, $\dampl$ and $\dampr$ allow particles encoded by the first information bit (`\matter particles') to clump together by emitting particles encoded by the second information bit (`\heat particles'). Conversely, when an \agglomeration of \matter gets hit by a \heat particle, the radiation particle is absorbed, and a \matter particle separate. This is analogous to an inelastic shock, where the collision of two massive objects lets the kinetic energy due to their relative velocity dissipate into heat. Similarly, radiation will erode an object. We write $\damp := \dampl\dampr$.

Next, we combine inelastic shocks with our previous ability to shrink/expand the graph, just by placing ourselves in ${\cal C}_3$  i.e. adding a $3$rd bit of information per port, on which the toy model dynamics $\hs\inter$ operates. In order not to interfere with other particles types, we just restrict $\inter$ to apply if and only if the neighbourhood does not contain \matter nor \heat particles, and call this $\interc$. Fig. ~\ref{fig:diagram_espace_entropie_SH} shows the spacetime diagram starting from a configuration containing no \heat particle, and density $\frac{1}{4}$ for the other particles. Fig.~\ref{fig:diagram-espace-chaleur-joli} instead starts from a configuration saturated with \heat particles: notice how it ends up with fewer \matter clumps.

\begin{figure}\centering
\begin{subfigure}{0.95\columnwidth}
\centering
\includegraphics[width=0.95\columnwidth]{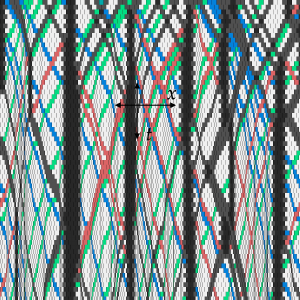}
\caption{~}
\label{fig:diagram_espace_entropie_SH}
\end{subfigure}\\
\begin{subfigure}{0.95\columnwidth}
\centering
\includegraphics[width=0.95\columnwidth]{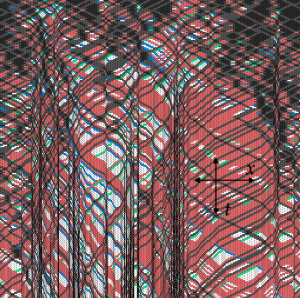}
\caption{~}
\label{fig:diagram-espace-chaleur-joli}
\end{subfigure}\hfill
\caption{\emph{Spacetime diagram of dynamics $\hs\inter\damp$.} Matter particles are shown in levels of grey, radiation particles are shown in levels of red. The particles responsible for expansion are shown in green and blue. \corrA{The point of this more elaborate dynamics is that, in spite of expansion, matter does not dilute so quickly, as if forms conglomerates. Thermal death is slower.}}
\end{figure}

\paragraph*{Conservation law}
The following conservation law establishes a relation between the size of the \matter clumps, and the level of \heat.

\newcommand{\energie}{E}
\newcommand{\energieb}{E'}
\newcommand{\bitv}[3]{#1_{#2}(#3)}

\begin{theorem}\label{th:loiconservationchoc}
For all $\gcirc \in \setcycle{2}$, $\energie(\hs\damp(\gcirc)) = \energie(\gcirc)$ where $\energie$ is defined by :
 \begin{align}  \energie(\gcirc) &:= \!\!\!\sum_{u \in V(\gcirc)} \!\!\left(\bitv{a}{1}{u}\bitv{a}{1}{u.b}\bitv{b}{1}{u.b}+\bitv{b}{1}{u}\bitv{a}{1}{u.a}\bitv{b}{1}{u.a} \right)\\ 
 &-\sum_{u \in V(\gcirc)} \left(\bitv{a}{2}{u}+\bitv{b}{2}{u}
\right)\end{align}
i.e. the number of \matter particles that are such that the neighbour on the opposite port contains two \matter particles, minus the number of \heat particles, is constant.
\end{theorem}
\begin{proof}(Outline). The proof technique is step-by-step and case-by-case keeping track of $\energie$, thus akin to a program-invariant proof. We shift it to Appendix~\ref{sec:loiconservationchocproof} for readability.  
\end{proof}
\VL{{\em Details on the model.} From this conservation law we deduce that the number of particles can only ever vary within a factor of two with respect to that present initially. This provides a bound on the sum of local entropies.
\begin{corollaire}\label{cor:borned_sum_entropy}
We have that for all $\gcirc \in \setcycle{3}$, and for all $n \in \N$, $ p\left(\gcirc\right)/2\leq p\left((\hs\interc\damp)^n (\gcirc)\right)\leq 2\times p\left(\gcirc\right)$, where $p(X)$ is the number of particles of $X$.
Consequently, $\sommeent{r}{\hs\interc\damp)^n (\gcirc)} \leq \sum_{i=0}^{2 \times p(X)} \log(\comb{i}{2r+1}) = 2 \times p(X) \times \log(2r+1)$. %
\end{corollaire}
From this lower and upper bound on the number of particles it also follows from Th.~\ref{th:entropywithsize} and Cor.~\ref{cor:thermaldeath} that for any graph $\gcirc$ such that $\left(|V((\hs\interc\damp)^n (\gcirc))|\right)_{n \in \N}$ is strictly increasing, the global entropy tends to $+\infty$ and the average local entropy tends to $0$.}

\subsection{Trials}

\begin{figure}[b]
\centering
\begin{subfigure}{0.95\columnwidth}\centering
\includegraphics[width=0.95\columnwidth]{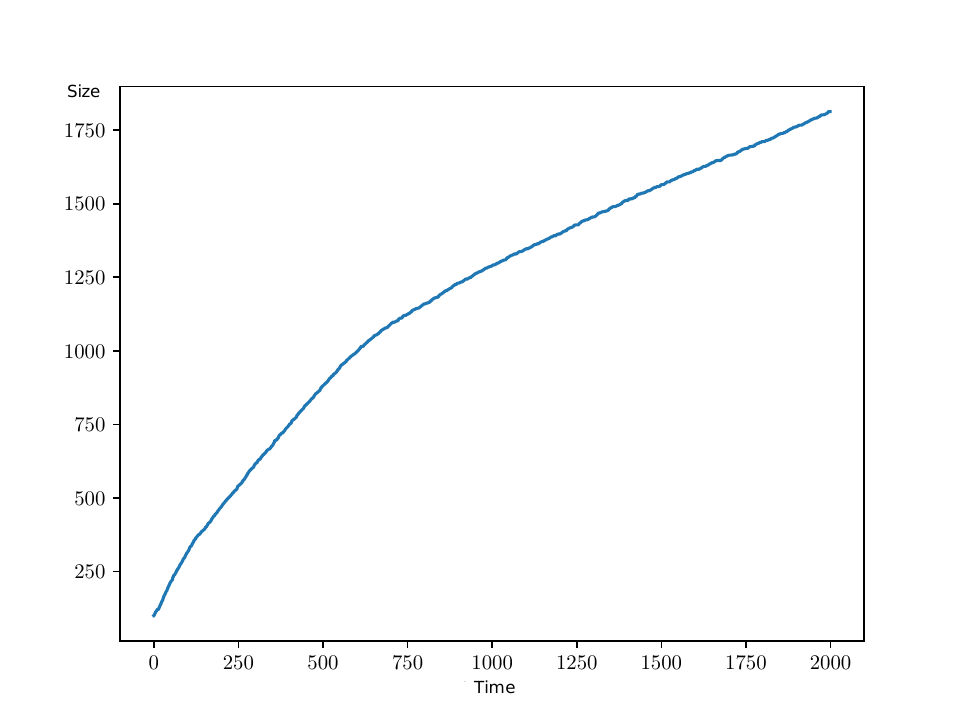}
\caption{\emph{Size of the graphs.}}
\label{fig:evolution-typique-taille-sid}
\end{subfigure}\\
\begin{subfigure}{0.95\columnwidth}\centering
\includegraphics[width=0.95\columnwidth]{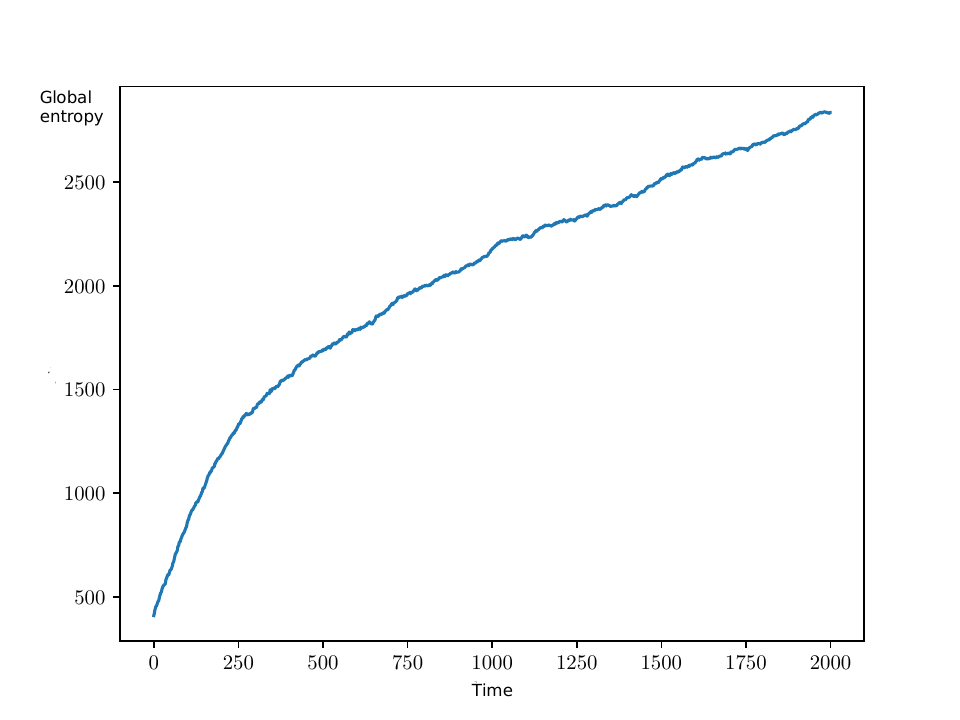}
\caption{\emph{Global entropy of the graphs.} }
\label{fig:global_entropy_sid}
\end{subfigure}
\caption{\emph{Typical size and global entropy curves for dynamics $\hs\inter\damp$.} The initial condition is drawn uniformly at random amongst graphs of size $100$.}
\end{figure}

\begin{figure}
\begin{subfigure}{0.95\columnwidth}
\centering\hfill
\includegraphics[width=0.95\columnwidth]{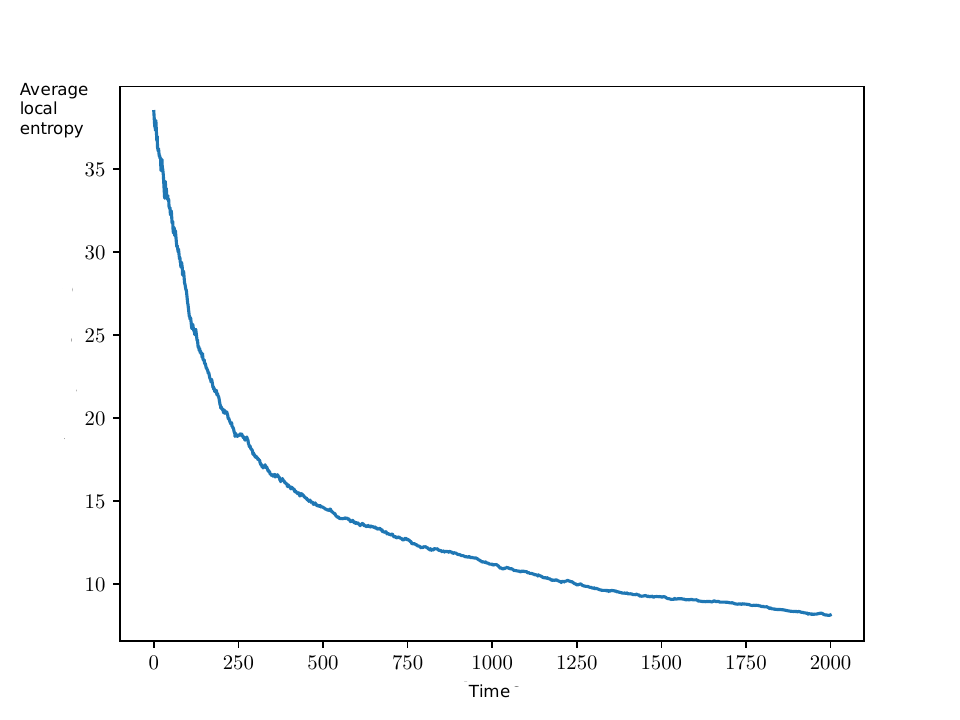}
\caption{\emph{Average local entropy of the graphs.}}
\label{fig:avg_entropy_sid}
\end{subfigure}\\
\begin{subfigure}{0.95\columnwidth}\centering
\includegraphics[width=0.95\columnwidth]{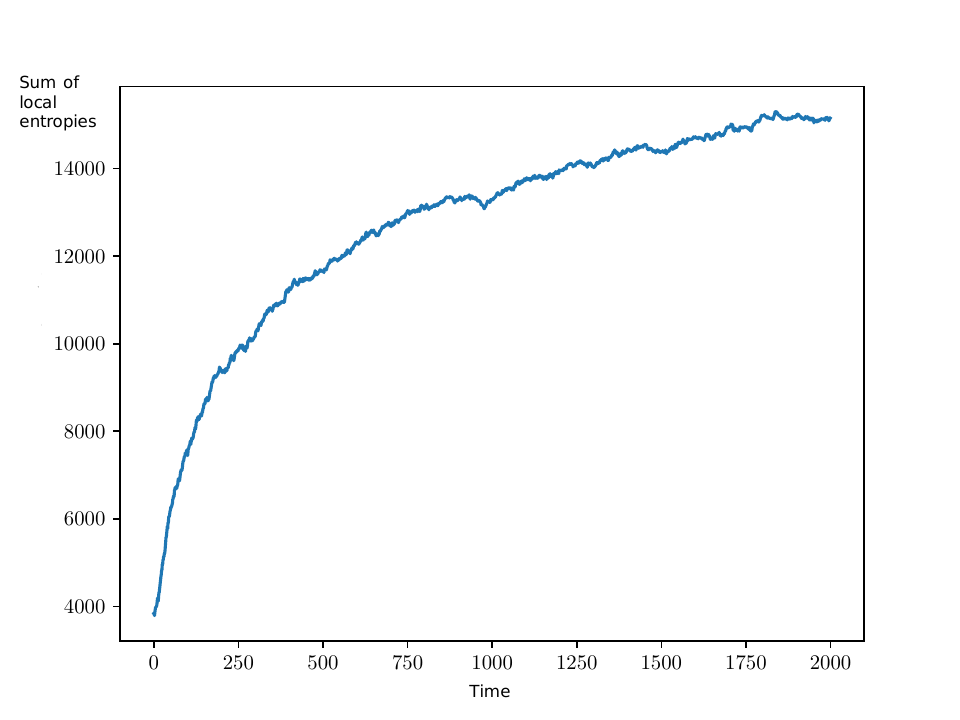}
\caption{\emph{Sum of the local entropies of the graphs.} }
\label{fig:space_time_entropy_1}
\end{subfigure}
\caption{\emph{Typical local entropy curves for dynamics $\hs\inter\damp$.} The initial condition is drawn uniformly at random amongst graphs of size $100$.}
\end{figure}

Again the complexity of $\hs\damp$ comes at a price, with potentially $3$ types of paticles per nodes interactions through step $\interc$ are sometimes prevented in a way that makes it difficult to extend Th.~\ref{th:sqrtSI_growing}. Still, we again have strong numerical evidence that generic graphs expand, with a size curve of order $\orderof(\sqrt{n})$ matching that of the underlying dynamics of $\hs\inter$ that is driving the growth, cf. Fig.~\ref{fig:evolution-typique-taille-sid}. Typically, other particles seem to only slightly slow down this growth, and may only totally prevent it in extreme cases. Again, as a result of this growth, global entropy grows (Fig.~\ref{fig:global_entropy_sid}) and average local entropy dies out (Fig.~\ref{fig:avg_entropy_sid}). Again the sum of the local entropies reaches a plateau (Fig. ~\ref{fig:space_time_entropy_1}). Now, the interesting point is that it reaches is much more slowly. 

The origin of this slower growth can be understood thanks to the conservation law of Th.~\ref{th:loiconservationchoc}. The inelastic shocks create \matter clumps of not-so-low density, where enough degrees of freedom survive so as to preserve the possibility for local entropy growth within these matter-populated regions. 
This however, by the conservation law, is done at the cost of the emitting \heat particles (which increases the entropy of other local windows). The formation of a clump of matter is therefore temporary, lasting only as long it is not collided by a emitted \heat particles (cf. Fig.~\ref{fig:diagram-espace-chaleur-joli}).
But as the graph size grows, \heat density lowers, and so the \heat travel times are longer. It follows that \matter clumps survive for longer and longer, thereby preserving the possibility for local entropy growth for longer and longer.

This mechanism gets amplified when replacing the underlying growth dynamics $\hs\inter$, with $\hs^{-1}\inter^{-1}$, or with exponentially growing dynamics like $\hs\interd$ and $\ecp\hs\inter$. For such dynamics, emitted \heat particles typically never make it back to the \matter clumps, which act as stable `entropy reservoirs'. 

\section{Conclusion}\label{sec:conclusion}

\paragraph{Underlying physical assumptions.} 
\corrA{This paper is motivated by the understanding that Physics Laws are fully time-reversible, and that this contradicts our everyday experience of time irreversibility---thereby calling for an explanation. Whilst time-symmetry is manifest in classical mechanics, general relativity and `pure' quantum mechanics, it clearly fails as soon as quantum measurements are introduced, as these are non-unitary. Whilst this paper is not cast in a quantum theoretical language, it is quantum-theory-ready, as the reversible evolutions considered can all be made into unitary over the Hilbert space of superposition of configurations by sheer linearity. Doing so describes a toy `single quantum wave universe'. Quantum measurements could then be recovered in a more epistemic fashion just by tracing out certain subsystems, in the spirit of decoherence theory \cite{PazZurek}.
}

\corrB{This paper is also motivated by the understanding that the Past hypothesis is too strong an assumption, because ``low entropy states are unlikely amongst the set of all states''. The sense in which they are unlikely, however, could be discussed. For instance, if the Past hypothesis was to be elevated to the status of a Physics Law, then they would not be unlikely. Notice, however, that Physics Laws are generally about prescribing a dynamics over a state space in a way that is homogeneous in space and time---not about positing the likelihood of some state over the other at a particular time. The Past hypothesis law would be quite a peculiar one in that sense. Alternatively, the Past hypothesis can be understood along the lines of Bayesian credence \cite{BayesianMyrvold}, i.e. the idea that based upon our observations about the universe, it likely started in a low entropy state. This is not incompatible with the conclusions of this paper, which seek to explain the reasons for these observations. 
}

\paragraph{Analytical results.} This paper provides a first rigorous proof that an arrow of time \corrA{can} emerge from almost all configuration under some time-reversible dynamics, without the need for a past hypothesis. \corrA{It does so by means of a toy model,} as the reversible dynamics in question is cast in the setting of reversible causal graph dynamics ; the methods used pertain to the field of theoretical Computer Science. The proof works by showing that graph size increases, and that entropy increases with size, where entropy is defined as for perfect gases. It provides a local explanation for the origin of the arrow of time, by tracing it back to local, reversible expansion mechanisms acting over configurations of finite but unbounded size.\\
This explanation resolves two of the three main drawbacks to the standard of the standard Boltzmann argument: there is no recurrence time, and no need to assume atypical initial conditions. One of the criticism still holds as there are successive configurations of minimal entropy from which two arrows of time flow in opposite directions. This idea is present in some cosmological models and popularized under the name or \lq Janus point\rq.

\paragraph{Numerical results.}
We also provide evidence, though computer simulations, that the explanation is compatible with fully time-symmetric dynamics, as well as periods of exponential growth.\\ 
In Thermodynamics, entropy increases globally, but locally it may well decrease and stabilize close to zero, reaching the so-called `thermal death'. At which stage it becomes impossible to witness an arrow of time locally. 
This is expected, and does occur in the toy model. We in fact provide a theorem about the fatality of thermal death in such models. Indeed as the graph expands, global entropy increases without bound, because ``there are more ways in which to position the matter'', but as the same time matter dilutes. In any fixed-sized window, matter becomes scarce, there are less ways to position it, and so local entropy decreases. The sum of the local entropies still increases as it should, but soon it stabilizes around an upper bound, characteristic of thermal death. All too quickly, there is no longer enough matter for the local entropy to increase; there is nowhere in the graph where we are able to observe an arrow of time locally. This is somewhat unrealistic. Although our universe, $\sim 13.7$ billion years after the big bang, is indeed largely empty, it still has regions where enough matter is concentrated so that local entropy may continue to increase, allowing us to observe local entropic clocks. According to Reichenbach and Rovelli \cite{RovelliBackToReichenbach, RovelliBackToReichenbach2} this is mainly due to the clumping of matter and the emission of radiation, through a delicate interplay between nucleosynthesis and gravity: modelling such processes lied way beyond the scope of this paper and therefore belongs to the perspectives. Yet, we do provide numerical evidence that just by adding inelastic collisions of matter to the model (and triggering the emission of radiation into free space so as to keep things time-reversible) the sum of the local entropies increases way slower. Locally, we are able to observe an arrow of time for longer. Thermal death is delayed.\\
\corrA{
\paragraph{Relevance of the toy models, to actual Physics.}
The use of toy models is ubiquitous in Physics. Indeed, for any explanation of any phenomenon we are better off within a simplified model, than using the full blown standard model over a general relativistic background.\\
Yet, some toy models are more realistic than others. Here is what can be said about the ones developed in this paper. 
\begin{itemize}
    \item The particles considered have spin, but no mass. It would be easy to add mass and thereby ensure that these particles propagate like fermions in the appropriate continuum limit \cite{arrighi2013dirac}.
    \item The geometry of the toy models is somewhat analogous to a Roberston-Walker metric, with quite noticeable differences: $1+1$ dimensions instead of $3+1$; discretized; and with a space-inhomogeneous scale factor. 
    \item Particle crossings trigger the spatial expansion (or sometimes shrinking), analogous to variations of the $g_{xx}$-component of the metric. Whilst this is likely the simplest reversible expansion mechanism one can imagine, it clearly is not in line with the Einstein field equations.
\end{itemize}
Although the four variants seek to reproduce other physically relevant features, they certainly do not constitute realistic models of expansion.\\
Yet, all five models, and in spite of their variations, constitute proofs of concepts that reversible expansion mechanisms lead to time arrow without past hypothesis. This is relevant to Physics per se, because this robustness suggests that most reversible expansion mechanism; including those stemming from general relativity, could have the same effects.\\
Thus, this toy model explanation may constitute a simple yet rigorous illustration of the core logic that drives the time arrow at the cosmological level. 
}

\paragraph{Perspectives.} We wonder whether the similar arguments can be made based on other notions of entropy, such as: metric entropy, topological entropy, or Von Neumann entropy (in the quantum regime), as these seem less dependent on choices of macroscopic variables. More crucially perhaps, we wonder whether such models can be brought to 2D and 3D.

\subsection*{Acknowledgements}

We wish to thank Marios Christodoulou, Pierre Guillon, Benjamin Hellouin, Carlo Rovelli and Francesca Vidotto for many discussions and insights. We also wish to thank the anonymous referees for their very helpful comments. This project/publication was made possible through the support of the ID\# 62312 grant from the John Templeton Foundation, as part of the \href{https://www.templeton.org/grant/the-quantum-information-structure-of-spacetime-qiss-second-phase}{‘The Quantum Information Structure of Spacetime’ Project (QISS)}. The opinions expressed in this project/publication are those of the author(s) and do not necessarily reflect the views of the John Templeton Foundation.

\bibliography{biblio,zotero,perso}

\appendix

\section{Named graphs}\label{sec:namedgraphs}

Say as in Fig.~\ref{fig:intuitive_I} that some quantum evolution splits a vertex $u$ into two. We need to name the two infants in a way that avoids name conflicts with the vertices of the rest of the graph. But if the evolution is locally-causal, we are unable to just `pick a fresh name out of the blue', because we do not know which names are available. Thus, we have to construct new names locally. A natural choice is to use the names $u.l$ and $u.r$ (for left and right respectively). Similarly, say that some other evolution merges two vertices $u,v$ into one. A natural choice is to call the resultant vertex $u\lor v$, where the symbol $\lor$ is intended to represent a merger of names.

This is, in fact, what the inverse evolution will do to vertices $u.l$ and $u.r$ that were just split: merge them back into a single vertex $u.l\lor u.r$. But, then, in order to get back where we came from, we need that the equality $u.l\lor u.r=u$ holds. Moreover, if the evolution is time-reversible, then this inverse evolution does exists, therefore we are compelled to accept that vertex names obey this algebraic rule.

Reciprocally, say that some evolution merges two vertices $u,v$ into one and calls them $u\lor v$. Now say that some other evolution splits them back, calling them $(u\lor v).l$ and $(u\lor v).r$. This is, in fact, what the inverse evolution will do to the vertex $u\lor v$, split it back into $(u\lor v).l$ and $(u\lor v).r$. But then, in order to get back where we came from, we need the equalities $(u\lor v).l=u$ and $(u\lor v).r=v$.

\begin{definition}[Names]\label{def:namealgebra}
Let $ \mathbb{K}$ be a countable set.
The name algebra $ \mathcal{N}[\mathbb{K}]$ has terms given by the grammar 
\begin{equation*}
u,v\ ::=\ c\ |\ u.t\ |\ u\lor v\quad\text{with} \quad c\in \mathbb{K} ,\ t\in \{l,r\}^{*}
\end{equation*}
and is endowed with the following equality theory over terms (with $ \varepsilon $ the empty word):
\begin{equation*}
( u\lor v) .l=u\qquad ( u\lor v) .r=v\qquad u.\varepsilon =u\qquad u.l\lor u.r=u
\end{equation*}
We define $ \mathcal{V}:=\mathcal{N}[\mathbb{K}]$. 
\end{definition}
The fact that this algebra is well-defined was proven in \cite{ArrighiCreation}.  Now that we have a set of possible names for our vertices, we can readily define `port graphs' (aka `generalized Cayley graphs' \cite{ArrighiCayleyNesme}). Condition \eqref{eq:wellnamedness} will just ensure that names do not intersect, e.g. forbidding that there be a name $u\vee v$ and another $v\vee w$, so as to avoid name collisions should they split.
\begin{definition}[Named graphs]\label{def:graphs}
Let $ \Sigma $ be the set of internal states and $\pi$ be the set of ports. A graph $G$ is given by a finite set of vertices $V_G\subseteq \mathcal{V}$ such that for all $v,v'\in V_G$ and for all $t,t'\in \{l,r\}^{*}$,
\begin{align}
v.t=v'.t'\textrm{ implies } v=v' \textrm{ and } t=t'\label{eq:wellnamedness}
\end{align}
together with 
\begin{itemize}
\item $\sigma_G : V_G\to\Sigma $ its internal states
\item $E_G$ a set of non-intersecting two element subsets of $V_G:\pi$, its edges.
\end{itemize}
In other words an edge $e$ is of the form $\{x : a, y : b\}$ and $\forall e,e'\in E_G., e\cap e'\neq \emptyset \Rightarrow e=e'$.
\end{definition}

\section{Proof of Th.~\ref{th:sqrtSI_growing}}\label{sec:expansionproof}

\newcommand{\ind}[1]{\indice(#1)}
\newcommand{\sumofv}[1]{\sum_{i=0}^{#1}v_i}
\begin{proof}
First, it is argued that there is a time step $m$ after which no more vertex mergers will occur. We denote $n_f(\gcirc)$ the number of merger patterns in $\gcirc$, and $d_{f}(\gcirc)$ the minimum distance between a merger pattern and a particle moving towards it (this includes a particle present in a merger pattern, and itself if there are no other particles).  We will show that the pair $\left(n_{f}({\hs } \inter)^n (\gcirc), d_{f}({\hs } \inter)^n (\gcirc)\right)$ decreases strictly in lexicographic order.

As we have seen in Lem.~\ref{lem:fusion_instability}, a merger pattern cannot be created, and is destroyed on collision. We only need to prove that $d_{fu}((\hs \inter)^n \gcirc)$ decreases strictly when $p\left({\hs } \inter)^n \gcirc\right)$ remains constant. Let $u,v$ be two vertices of a merger pattern and $p$ a particle such that $u,v$ and $p$ realise the distance $d_{f}(\gcirc)$. Two cases can occur, either $p$ is itself part of a merger pattern, in which case there are no particles between the two merger patterns, or $p$ is free moving, in which case there are only particles going in the opposite direction between $p$ and $u,v$. In the first case, the two perform a fusion and there are no particles between the two merger patterns (so there is no division); the distance between the two patterns therefore decreases by $1$. In the second case, the particle $p$ will move towards the merger pattern, decreasing $d_{fu}((\hs \inter)^n \gcirc)$.

\sloppy In order to preserve the readability of the notations, we will denote $(u_n)$ the sequence $\left(|V((\hs \inter)^n(\gcirc))|\right)$. Thanks to the previous point, we know that there is a time step $m\in \N$ from which each collision of particles will cause the creation of an additional vertex, so the sequence $u_n$ is necessarily increasing for all $n>m$. Since $\gcirc$ contains at least one particle of each type, we have that for all $n \geq m$, the evolution of $(\hs \inter)^n \gcirc$ during $u_n$ time steps causes at least one collision. Similarly, we know that at most $c=2n_an_b$ collisions occur in the same time frame where $n_a$ (resp. $n_b$) is the number of particles on the port $a$ (resp. $b$). This allows us to obtain the following inequalities:
\begin{equation}
 u_n + 1 \leq u_{n+u_n} \leq u_n + c
\label{main_eq}
\end{equation}

Let $(v_k)_{k \in \N}$ be the sub-sequence such that $v_0=u_m$ and for all $k \in \N$, $v_{k+1}=u_{\ind{k} + v_k}$, where $\ind{k}$ is the function such that $\ind{0}=m$ and for all $k \in \N, \ind{k}=\sumofv{k-1}$.
By induction, we prove that $u_{\ind{k}}=v_k$ :
\begin{equation}
u_{\ind{k+1}} = u_{\sumofv{k}} =  u_{v_k + \sumofv{k-1}} = u_{\ind{k} + v_k} = v_{k+1}
\label{indice_eq}
\end{equation}
This allows us to apply the inequality \eqref{main_eq} on $v_{k+1}=u_{\ind{k} + v_k}$. Combining \eqref{main_eq} and \eqref{indice_eq}, we obtain the linear growth of $(v_k)_{k \in \N}$:
\begin{equation}
v_k+1 \leq v_{k+1} = u_{\ind{k}+v_k} \leq v_k+c
\label{v_growing_eq}
\end{equation}

Let us now focus on the growth of the index of $(v_k)_{k \in \N}$. By applying the inequalities of \eqref{v_growing_eq} to the definition of $\ind{k}$ we obtain :

\begin{align}
\sum_{i=0}^k (v_0+i) &\leq \sum_{i=0}^k (v_i) \leq \sum_{i=0}^k (v_0 + ci) \\
ku_m + \frac{k(k-1)}{2} &= \sum_{i=0}^k (v_0+i) \leq \ind{k} \\
&\leq \sum_{i=0}^k (v_0+ci) = ku_m + c\frac{k(k-1)}{2} \label{quadratic_ind_eq}
\end{align}

To conclude, let us return to the main sequence $(u_n)_{n \in \N}$. For a sufficiently large $n$, there exists $k \geq 4u_m + c$ such that :
\begin{align} \ind{k} \leq n \leq \ind{{k+1}}\end{align}
This gives us, considering that $(u_n)_{n\in \N}$ is increasing, the previous equation \eqref{quadratic_ind_eq} and that $k \geq 4u_m + c$ the following inequalities:
\begin{align}
  \ind{k} &\leq n \leq \ind{{k+1}} \\
\implies  ku_m + \frac{k(k-1)}{2}  &\leq n \leq (k+1)u_m + c\frac{k(k+1)}{2} \\
\implies  k(2u_m + k - 1) &\leq 2n \leq k(4u_m + c + ck)\\ 
 \implies  2k^2 &\leq 2n \leq (c+1)k^2\\ 
\implies  k &\leq \sqrt{n} \leq \sqrt{\frac{(c+1)}{2}}k\\ 
\implies   \sqrt{\frac{2}{c+1}}\sqrt{n}&\leq k \leq \sqrt{n} \label{k_bound_eq}
\end{align}

Since the sequence $u_n$ is increasing, and using the inequalities \eqref{k_bound_eq} and \eqref{v_growing_eq}, we can conclude with the following inequalities:
\begin{align}  u_m + \sqrt{\frac{2}{c+1}}\sqrt{n} &\leq u_m + k \leq v_k \leq u_n \\
&\leq v_{k+1} \leq u_m + c(k+1)\\
&\leq u_m + c + c \sqrt{n} \end{align}

Thus, $u_n$ is of the order of $(\sqrt{n})$.
\end{proof}

\section{Proof of Th.~\ref{th:loiconservationchoc}}\label{sec:loiconservationchocproof}

\begin{proof}
The function $\energieb$ is defined by : \begin{align}  \energieb(\gcirc) &:=\!\!\!\sum_{u \in V(\gcirc)} \!\!\bitv{a}{1}{u}\bitv{a}{1}{u.a}\bitv{b}{1}{u.a}+\bitv{b}{1}{u}\bitv{a}{1}{u.b}\bitv{b}{1}{u.b}\\ 
&-\sum_{u \in V(\gcirc)} \left(\bitv{a}{2}{u}+\bitv{b}{2}{u}
\right)\end{align}

This time, it amounts to counting the number of matter particles such that the neighbour on the same port contains two matter particles and deducing the number of \heat particles. Following this intuition, we introduce the localised energy functions $\energie$ and $\energieb$ such that for all $\gcirc \in \setcycle{2}$ and for all $u\in \gcirc$ :
\begin{align} \energie(\gcirc,u) &= \bitv{a}{1}{u}\bitv{a}{1}{u.b}\bitv{b}{1}{u.b}\\ 
&+ \bitv{b}{1}{u}\bitv{b}{1}{u.a}\bitv{a}{1}{u.a} \\
&- \bitv{a}{2}{u} -  \bitv{b}{2}{u} \end{align}
   and
   \begin{align} \energieb(\gcirc,u) &= \bitv{a}{1}{u}\bitv{a}{1}{u.a}\bitv{b}{1}{u.a}\\ &+ \bitv{b}{1}{u}\bitv{b}{1}{u.b}\bitv{a}{1}{u.b}\\ &- \bitv{a}{2}{u} -  \bitv{b}{2}{u} \end{align}
This notation is justified by the equalities:
\begin{align}  \energie(\gcirc) := \sum_{u \in V(\gcirc)}  \energie(\gcirc,u)\end{align}
and
\begin{align} \energieb(\gcirc) := \sum_{u \in V(\gcirc)} \energieb(\gcirc,u)\end{align}
Let us prove that for all $\gcirc$, we have $\energie(\gcirc) = \energieb(\damp(\gcirc)) = \energie(\hs\interc\damp(\gcirc)) $.

First, let us prove that $\energie(\gcirc,u)=\energieb(\damp\gcirc,u)$. By definition, $\damp$ acts only in the neighbourhood of vertices $u$ such as $\bitv{a}{1}{u}=\bitv{b}{1}{u}=1$. Moreover, when $\dampl$ and $\dampr$ act on the same vertex $v$, we have by symmetry of the ports $a$ and $b$ that the vertex $i$ is left unchanged by $\damp$. By the same symmetry, we can assume without loss of generality that only $\dampl$ affects vertex $u$, which gives us the following calculations:
\newcommand{\scaleem}{2}
\begin{widetext}
$$
\begin{array}{rccrl}
E\left(\raisebox{-10pt}{\includegraphics[scale= \scaleem]{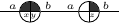}}\right) &=& 1 - (x+y+z) &=& E'\left(\raisebox{-10pt}{\includegraphics[scale= \scaleem]{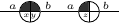}}\right)\\

E\left(\raisebox{-10pt}{\includegraphics[scale= \scaleem]{figures/Time_arrow/energy_conservation_proof/motif2.pdf}}\right)&=&-(x+y+z)&=&E'\left(\raisebox{-10pt}{\includegraphics[scale= \scaleem]{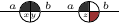}}\right)\\

E\left(\raisebox{-10pt}{\includegraphics[scale= \scaleem]{figures/Time_arrow/energy_conservation_proof/motif4.pdf}}\right) &=& -(x+y+z+1) &=& E'\left(\raisebox{-10pt}{\includegraphics[scale= \scaleem]{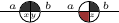}}\right)\\

E\left(\raisebox{-10pt}{\includegraphics[scale= \scaleem]{figures/Time_arrow/energy_conservation_proof/motif3.pdf}}\right) &=& -(x+y+z) &=& E'\left(\raisebox{-10pt}{\includegraphics[scale= \scaleem]{figures/Time_arrow/energy_conservation_proof/motif1.pdf}}\right)\\

\end{array}$$

In summary, we have that $E(\gcirc) = E'( \damp (\gcirc))$. Similarly, for $\hs$ we get :

$$
\begin{array}{rcl}
E'\left(\raisebox{-10pt}{\includegraphics[scale=\scaleem]{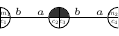}}\right)&=& m_1+m_2 - (c_1 + c_2 + c_3 + c_4)\\
 &=& E\left(\raisebox{-10pt}{\includegraphics[scale=\scaleem]{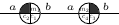}}\right)
\end{array}$$

And so we get that $E'(\gcirc)=E(\hs(\gcirc))$. Combining all this, we obtain that $E(\gcirc) = E'( \damp (\gcirc))=E(\hs\damp(\gcirc))$
\end{widetext}
\end{proof}

\section{Variable radius entropy}\label{sec:variableentropy} 

To still see a growth in local entropy, in spite of expansion, we need to compensate for the loss of particle density. There are two simple ways to do this. The first is to directly increase the radius of observation in proportion to the size of the graph, so that each disk always represents the same proportion of the graph, which will preserve the particle density. %
The other way is to use the names as markers to split the configurations into different induced subgraphs.

For the first case, it is sufficient to pose for $r \in \left[0,1\right]$ :
\begin{align} S'_r(X) := \frac{1}{|V(X)|}\sum_{u\in V(X)} S(X_u^{\lceil r\times |V(X)| \rceil})\end{align}
This gives us an increasing entropy in the case of $\hs\inter$, $(\hs\inter)^{-1}$ and $\hs\interb$ as can be seen in Fig. ~\ref{avg_entropy_variable_window_05}.

\begin{figure}\centering
\begin{subfigure}{0.95\columnwidth}
\centering
\includegraphics[width=0.95\columnwidth]{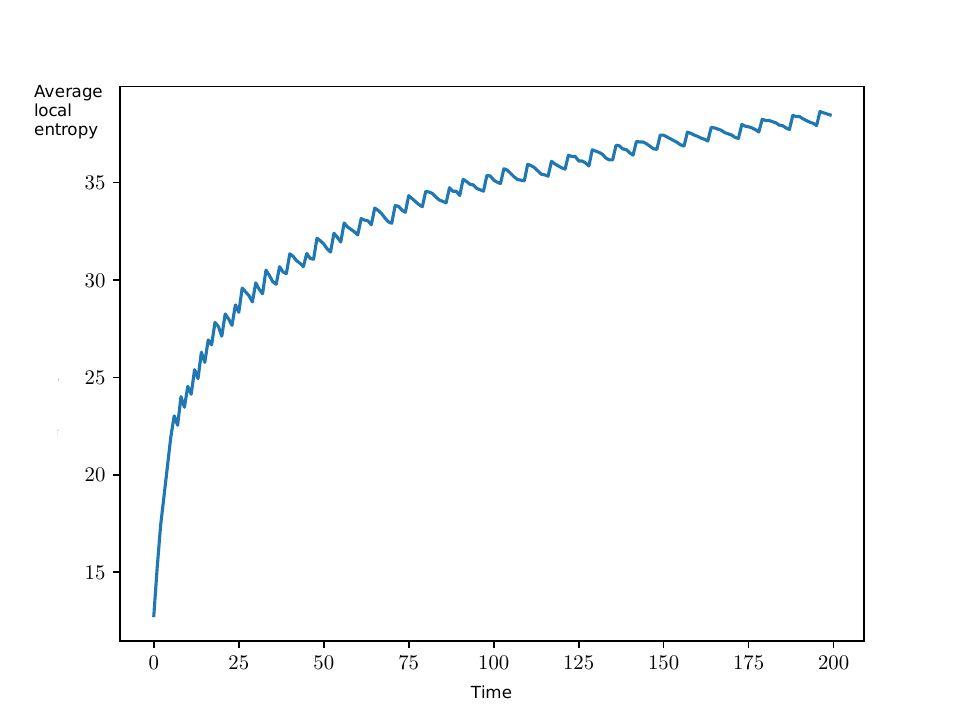}
\caption{\emph{Average $r=0.05$-local entropy of the graphs.}}
\label{avg_entropy_variable_window_05}
\end{subfigure}\\
\begin{subfigure}{0.95\columnwidth}
\centering
\includegraphics[width=0.95\columnwidth]{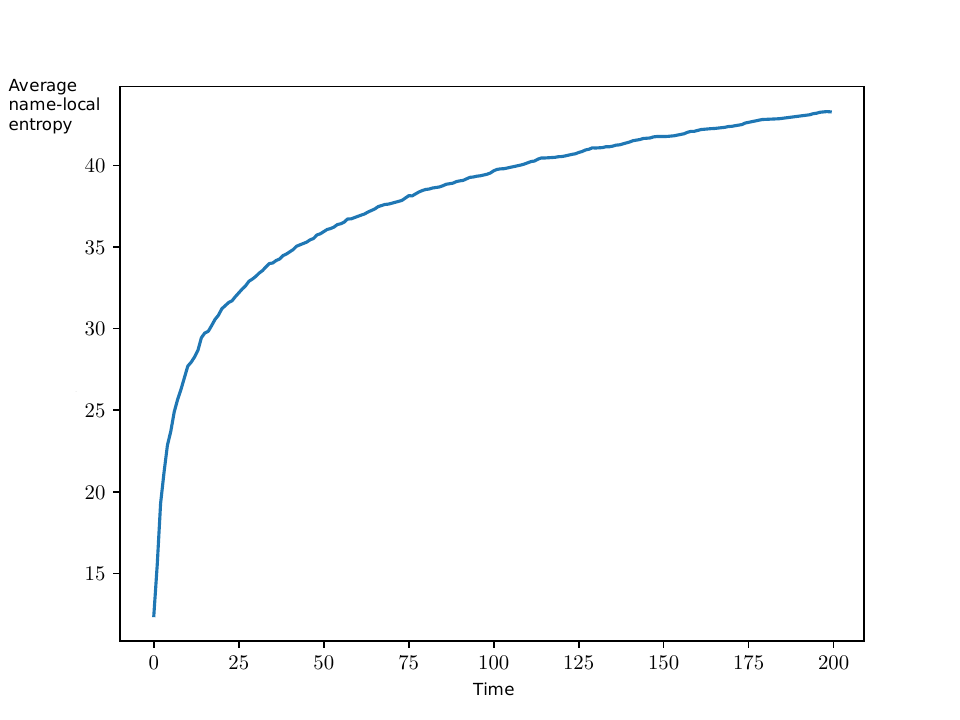}
\caption{\emph{Average name-local entropy of the graphs for windows of size $5$ initially.}}
\label{avg_entropy_variable_window}
\end{subfigure}
\caption{\emph{Typical of two modified average local entropy for the dynamics $\hs\inter$.} The initial condition is drawn uniformly at random amongst graphs of size $100$.}
\end{figure}

Although it compensates for the expansion of the graph, this entropy makes little physical sense: a node creation on the opposite side of the graph can increase the size of a local window. %
To get rid of this problem, it is sufficient to use the name algebra as a coordinate system, and use them to determine whether two vertices belong to the same window. As the names of the vertices are in the algebra $\mathcal{V}$, they can be seen as finite binary trees labelled by $\N.\bint$. It is therefore sufficient to state that a window consists of the set of vertices whose leftmost child intersects the same integer. More formally, for all $i \in \N$, we pose: \begin{align} V_i(X)=\{u|u\in V(X) \text{ et } \exists m \in \N  \text{ tel que } u.l^m \wedge i = u.l^m\}\end{align}

The average local entropy can then be defined as :
\begin{align} S_\mathcal{N}(X) := \frac{1}{|V(X)|} \sum_{i\in \N} S(X^0_{V_i}) \end{align}
 where $n$ is the number of integers composing the names of $X$ and $X^0_{V_i}$ is the subgraph induced by the set of vertices in $V_i(X)$.

This definition seems at first sight very restrictive, but it is important to remember that the dynamics of named graphs switch with any renaming. Thus, to observe how the local entropy evolves according to a different partitioning, we just need to rename the graph appropriately, and observe the entropy of the resulting graph.

Note that in this definition of entropy, local windows do not intersect. Moreover, a window only increases in size when the space inside it grows. This also makes physical sense to define entropy in this way: it is like placing a coordinate system in space and letting the dynamics distort it.

As for the previous definition of entropy, we always observe a local growth of entropy for the dynamics mentioned above, illustrated in Fig. ~\ref{avg_entropy_variable_window}.

\end{document}